\documentclass[12 pt]{amsart}
\usepackage{amsmath, amsthm, amsfonts, amssymb,mathrsfs,multirow,graphicx}
\theoremstyle{plain}

\newtheorem{thm}{Theorem}[section]
\newtheorem{rem}[thm]{Remark}
\newtheorem{remark}{Remark}[section]
\newtheorem{lemma}[thm]{Lemma}
\newtheorem{prop}[thm]{Proposition}
\newtheorem{cor}[thm]{Corollary}

\def\tr{{\mathop {{\rm tr}}}}
\newcommand\Z{\mathbb Z}
\newcommand\R{\mathbb R}

\newcommand\N{\mathbb N}
\makeatletter 
\@addtoreset{equation}{section}
\makeatother  

\begin{document}

\title{``Mixed spectral nature"    of   Thue--Morse Hamiltonian}

\author{Qinghui LIU}
\address[Q.H. LIU]{
Department of  Computer Science,
Beijing Institute of Technology,
Beijing 100081, P. R. China.}
\email{qhliu@bit.edu.cn}
\author{Yanhui QU}
\address[Y.H. QU]{Department of  Mathematical Science, Tsinghua University, Beijing 100084, P. R. China.}
\email{yhqu@math.tsinghua.edu.cn}
\author{Xiao Yao}
\address[X. Yao]{Department of  Mathematical Science, Tsinghua University, Beijing 100084, P. R. China.}
\email{yaox11@mails.tsinghua.edu.cn}

\begin{abstract}
We find three dense subsets $\Sigma_I,\Sigma_{II}$ and $\Sigma_{III}$ of the spectrum of the Thue-Morse Hamiltonian,
such that each energy in $\Sigma_I$ is extended;   each energy in $\Sigma_{II}$ is pseudo-localized and each energy in $\Sigma_{III}$ is one-sided pseudo-localized. We also obtain exact estimations on the norm of the transfer matrices and the norm of the formal solutions for these energies. Especially, for $E\in \Sigma_{II}\cup\Sigma_{III}, $  the norms of the transfer matrices behave like
$$
e^{c_1\gamma\sqrt{n}}\le\|T_{ n}(E)\|\le  e^{c_2\gamma\sqrt{n}}.
$$
The local dimensions of the spectral measure on these subsets are also studied.  The local dimension is $0$ for energy  in  $\Sigma_{II}$ and is larger than $1$ for energy in  $\Sigma_{I}\cup\Sigma_{III}.$
In summary,   the Thue-Morse Hamiltonian exhibits ``mixed spectral nature".

\end{abstract}

\maketitle

\section{Introduction}

Discrete Schr\"odinger operator  originates from crystal physics. To describe the motion of one electron in a one-dimensional  crystal, physicists use the following ``tight binding" model: $H_V:\ell^2(\Z)\to \ell^2(\Z)$ defined as
$$
(H_V\psi)_n= \psi_{n+1} +\psi_{n-1}+V(n)\psi_n,
$$
where $V:\Z\to \R$ is the potential. To understand the spectral property of $H_V$, the central issue  is to analyze the eigen-equation $H_V\psi=E\psi$, where $E\in \R$ is an energy. Define transfer matrices as follows:
$$
M_n(E):=
\begin{pmatrix}
E-V(n)&-1\\
1&0
\end{pmatrix};\quad
T_{m\to n}(E):=
\begin{cases}
M_{n}\cdots M_{m+1}& m< n\\
I&m=n\\
\left(M_{m}\cdots M_{n+1})\right)^{-1}& m>n.
\end{cases}
$$
Write $T_n(E):=T_{0\to n}(E)$ and  $\vec{\psi}_n:=(\psi_{n+1},\psi_n)^t$,
where $v^t$ is the transpose of $v$.
Then $H_V\psi=E\psi$ if and only if for any integers $m, n$,
\begin{equation}\label{transfer-vector}
\vec{\psi}_n=T_{m\to n} \vec{\psi}_m.
\end{equation}

\subsection{Literature and background}
For crystal, the potential $V$ is periodic.  Assume the period of $V$ is $N$, then  it is well-known that  the spectrum $\sigma(H_V)$ of $H_V$ satisfies
 \begin{equation*}
\sigma(H_V)=\{E\in \R: |\tr(T_{N}(E))|\le 2\},
\end{equation*}
which is made of finitely many intervals.  The spectral measure of $H_V$ is absolutely continuous. See for example \cite{CL,To}. From fractal geometry point of view, the spectral measure  has Hausdorff dimension $1$. In particular,  $T_N(E)$ is elliptic for any  $E\in \sigma(H_V)$ with $|\tr(T_{N}(E))|< 2$. As a consequence, there exists a constant $C>1$ such that, for any $n\in\Z$,
\begin{equation}\label{transfer-periodic}
1\le \|T_n(E)\|\le C
\end{equation}
 and for any solution $\psi$ of $H_V\psi=E\psi,$ any $n\in\Z$,
\begin{equation}\label{solution-periodic}
  C^{-1}\|\vec{\psi}_0\|\le \|\vec{\psi}_n\|\le C\|\vec{\psi}_0\|.
\end{equation}
In this case, $E$ is called an {\it extended} state.

In order to study crystal with impurities,
Anderson \cite{A} proposed the following model: assume $\{V_\omega(n): n\in \Z\}$ are i.i.d. random variables. For this model, the typical spectral picture is as follows.  Assume $\varphi$ is the probability distribution of $V_\omega(0)$, then almost surely
\begin{equation*}
\sigma(H_{\omega})={\rm supp}(\varphi)+[-2,2].
\end{equation*}
Thus the spectrum still has a band structure. On the other hand, almost surely, the spectral measure of $H_\omega$ is pure point; the set of eigenvalues $\{E^{\omega,i}:i\in \N\}$ is dense in $\sigma(H_\omega)$; the eigenfunctions $\{\psi^{\omega,i}: i\in \N\}$ form a basis of $\ell^2(\Z)$; each $\psi^{\omega,i}$ decays exponentially: there exists a constant $\gamma>0$ such that for any $\epsilon>0$ small,
\begin{equation}\label{exp-decay-random}
e^{-(\gamma+\epsilon)|n|}\lesssim \|\vec{\psi}^{\omega,i}_{ n}\|\lesssim e^{-(\gamma-\epsilon)|n|}.
\end{equation}
If $\phi$ is another solution of the eigen-equation and independent with $\psi^{\omega,i}$, then
\begin{equation}\label{exp-increase-random}
e^{(\gamma-\epsilon)|n|}\lesssim \|\vec{\phi}_n\|\lesssim e^{(\gamma+\epsilon)|n|}.
\end{equation}
Moreover, the following estimation on the norms of the transfer matrices holds:
\begin{equation}\label{transfer-random}
e^{(\gamma-\epsilon)|n|}\lesssim\|T_n(E^{\omega,i})\|\lesssim e^{(\gamma+\epsilon)|n|}.
\end{equation}
The energy $E^{\omega,i}$ is called a {\it localized } state. Since the spectral measure is supported on a countable set, it has Hausdorff dimension $0.$
The above property  is known as { \it Anderson localization}, see for example \cite{GMP,Mo,Pa,Ca,Ko,CKM,DSS}.

The periodic potentials and random potentials lie in two extremes according to the degree of orderness  of potentials.  In between, there are rich varieties.  The most popular potentials which have been extensively studied in the past three decades are the so-called {\it quasiperiodic} potentials. They all have their  roots in  physics, see for example \cite{AA,KKT,OPRSS,SBGC}.
   The spectral properties of the related operators exhibit big varieties.
Here we focus on two typical classes of potentials. One class consists of analytic potentials, the most famous example is the so-called Almost Mathieu potential defined by $V^{am}(n)= \cos 2\pi (n\alpha+\theta)$,
where $\alpha\in\R\setminus\mathbb{Q}$, $\theta\in\R$.
Another class consists of potentials generated by primitive substitutions, the most famous example is the so-called Fibonacci potential, defined by $ V^F(n)=\chi_{(\alpha,1]}(n\alpha+\theta\pmod 1)$, where $\alpha=(\sqrt{5}-1)/2$ is the inverse of Golden mean.

The spectrum  of Almost Mathieu operator $H_{\lambda V^{am}}$ is well understood now, it is known that the spectrum is a Cantor set (this is the famous Ten Martini problem, see \cite{AJ}) independent with $\theta,$ and has  Lebesgue measure
\begin{equation*}
|\sigma(H_{\lambda V^{am}})|=|4-2\lambda|
\end{equation*}
(see \cite{L,JK,AK}). For the spectral measure,
  Aubry and Andre \cite{AA} conjectured that for irrational frequency $\alpha$,   the spectral measure of the operator is absolutely continuous when $|\lambda|<2$, while it  is pure point with exponentially decaying eigenfunctions when $|\lambda|>2$. That is,  there exists a metal-insulator transition at $|\lambda|=2.$   It was  soon realized  that this conjecture is not true for certain frequencies $\alpha$ and phases $\theta$, see \cite{AS,JS}.  On the other hand, Jitomirskaya   \cite{J99} established  this conjecture for Lebesgue a.e. frequencies and phases. Consequently for such $\alpha$ and $\theta$, the eigenvalues, eigenfunctions and transfer matrices of the corresponding operator have similar properties as \eqref{exp-decay-random}, \eqref{exp-increase-random}  and \eqref{transfer-random}. The spectral measure still has Hausdorff dimension 0 since it is pure point. In general,  by subordinate theory, Jitomirskaya and Last \cite{JL00} showed that when $|\lambda|>2$, for any irrational $\alpha$ and $\theta,$ the spectral measure of related operator also has Hausdorff dimension 0, even not necessarily pure point. This means that  the spectral measure is still  close  to pure point in some sense. Recently Avila \cite{Av} established the absolute continuity of  spectral measure for any $\alpha$ and $\theta$ when $|\lambda|<2$.

  Now we turn to the second class. It is known that, if $V$ is defined through a primitive substitution, then
the spectrum $\sigma(H_{\lambda V})$ has  zero Lebesgue measure for any $\lambda\ne 0$,  see  \cite{BG,Lenz,LTWW}. Consequently the spectrum is a Cantor set. Remark that the Fibonacci potential can also be defined through the famous Fibonacci substitution $\tau$: $\tau(a)=ab$ and $\tau(b)=a$(see for example \cite{Fo} Sect. 5.4), thus $\sigma(H_{\lambda V^F})$ has Lebesgue measure 0. The Hausdorff dimension of $\sigma(H_{\lambda V^F})$ has been extensively studied, see \cite{R, JL00,LW,DEGT,DG,DG2}, especially the recent work \cite{DGY}. In particular,  the following property is shown in \cite{DEGT}:
\begin{equation}\label{asym-Fibonacci}
\lim_{|\lambda|\to\infty} \dim_H \sigma(H_{\lambda V^F})\ln  |\lambda|=\ln (1+\sqrt{2}).
\end{equation}
This implies that $\dim_H \sigma(H_{\lambda V^F})\to 0$   when $|\lambda|\to\infty.$ On the other hand, for any $E\in \sigma(H_{\lambda V^F})$
and any  solution $\psi$ of $H_{\lambda V^{F}}\psi=E\psi$,  by \cite{IT,IRT,JL00,DKL}, there exist $0<\alpha_1\le \alpha_2$ and $\beta>0$ such that
\begin{equation}\label{power-Fibonacci}
L^{\alpha_1}\lesssim \|\psi\|_L\lesssim L^{\alpha_2}\ \ \ \text{ and }\ \ \  \|T_n(E)\|\lesssim |n|^\beta,
\end{equation}
where $\|\psi\|_L$ is defined by
\begin{equation*}
\|\psi\|_L:=\left(\sum_{n=0}^{[L]}|\psi(n)|^2+(L-[L])|\psi([L]+1)|^2\right)^{1/2}.
\end{equation*}
Remark that \eqref{power-Fibonacci} is different from \eqref{exp-decay-random}, \eqref{exp-increase-random} and \eqref{transfer-random}, where for  the formal solution $\psi$,   $\|\psi\|_L$  is either bounded, or increasing   exponentially; the norms of transfer matrices also increase exponentially.
As a consequence of \eqref{power-Fibonacci},  the spectral measure $\mu$ is uniformly $\alpha$-H\"older continuous, where $\alpha=2\alpha_1/(\alpha_1+\alpha_2)$. This implies that  $d_\mu(E)\ge \alpha$ for any $E\in \sigma(H_{\lambda V^F})$(see \cite{DKL}), thus $\dim_H \mu\ge \alpha$. It is also known that $\dim_H\sigma(H_{\lambda V^F})<1$ for any $\lambda\ne 0$ (\cite{C,DGY}).  Since the spectral measure $\mu$ is supported on the spectrum, we have $\dim_H\mu <1.$
This implies that, in \eqref{power-Fibonacci}, $\alpha_1$ is strictly less than $\alpha_2$. Thus \eqref{power-Fibonacci} is also different from the periodic case, for which,  \eqref{transfer-periodic} implies that for any formal solution $\psi$,  $\|\psi\|_L\sim L^{1/2}$.

Now we consider Thue-Morse potential, which is the potential we will study in this paper.
It has  been  studied by many authors, see \cite{AAKMP,MBNP,RSL,AP1,AP,Lu,B,BBG2,BG,DT}.
The   Thue-Morse potential $w$ is defined   as follows:  Let $\varsigma$ be the Thue-Morse substitution  given by $\varsigma(a)=ab$ and $\varsigma(b)=ba$.   Let $u=u_1u_2\cdots:=\varsigma^\infty(a).$ For $n\ge 1$, let  $w(n)=1$ if $u_n=a$; let $w(n)=-1$
 if $u_n=b;$ let $w(1-n)=w(n)$ for $n\ge 1.$ The operator $H_{\lambda w}$
   is called the {\it Thue-Morse Hamiltonian}.  Since $\varsigma$ is primitive,  $\sigma(H_{\lambda w})$  has Lebesgue measure  0 when $\lambda\ne 0.$  On the other hand, it is shown in \cite{HKS,DGR} that the operator has no eigenvalues, consequently the spectral measure is purely singular continuous.  Recently it was shown in \cite{LQ} that  the Hausdorff dimension of $\sigma(H_{\lambda w})$ has an absolute positive lower bound for any $\lambda,$ which form a big contrast
  with \eqref{asym-Fibonacci}.

In this work, we concentrate on three subsets of the spectrum  $\sigma(H_{\lambda w})$  and study some finer  spectral properties of $H_{\lambda w}$.  More precisely, we  study the  behavior of the transfer matrices and  eigenfunctions for those energies,  we also study the local dimension of the spectral measure  at those energies.  From now on, we only consider the Thue-Morse Hamiltonian for $\lambda\ne 0$, thus we simplify the notation $H_{\lambda w}$ to $H_\lambda$.

\subsection{Main results}

We start with an estimation on the norm of the transfer matrices for any energy in the spectrum. It is known that(see e.g. \cite{BBG,BG,Lenz}),
for any potential generated by primitive substitution,
the Lyapunov exponent vanishes on the spectrum,
and it was implied by their proofs that,
there are $c=c(\lambda,E)>0$ and $0<\alpha<1$ such that, for any $n>0$,
$$\|T_n(E)\|\lesssim e^{cn^\alpha}, $$
where in the case of Thue-Morse, $\alpha>3/4$.
Here, we improve the exponent $\alpha$ to $1/2$.

\begin{thm}\label{main-subexp}
There exists $c=c(\lambda)>0$ such that, for any $E\in \sigma(H_\lambda)$ and $ n\in\N$,
\begin{equation}\label{uniform}
\|T_{\pm n}(E)\|\le e^{c n^{1/2}}.
\end{equation}
\end{thm}

One may ask, whether $\alpha=1/2$ is sharp, or equivalently, does  there exist some energy such that the norm of related transfer matrices increases at least as $e^{\tilde c\sqrt{n}}$ for some positive constant $\tilde c?$ We will answer this question later. Indeed, we will construct three subsets of the spectrum, such that among each of them, we can get more precise estimations on the norm of the transfer matrices.

Let us give a
description
on the structure of the spectrum. Define $t_n(E)$ to be the trace of $T_{2^n}(E)$.  It is a polynomial of order $2^n$ and is called the $n$-th {\it trace polynomial }related  to Thue--Morse Hamiltonian.
Define $\sigma_n:=\{E\in \R: |t_n(E)|\le 2\}$.
It is shown in \cite{LQ} that $\{\sigma_n\cup\sigma_{n+1}: n\ge 1\}$ is a decreasing set sequence and
\begin{equation}\label{structure-spectrum}
\sigma(H_\lambda)=\bigcap_{n\ge 1}(\sigma_n\cup \sigma_{n+1}).
\end{equation}

$E\in \R$ is called a {\it type-I} energy if $t_k(E)=0$ for some $k\in\N.$
$E\in \R$ is called a {\it type-II} energy if there exists some $n_0$ depending on $E$,
such that $E\in \sigma_{2n}\setminus \sigma_{2n+1}$ for all $n\ge n_0.$
$E\in \R$ is called a {\it type-III} energy if there exists some $n_0$ depending on $E$,
such that $E\in \sigma_{2n+1}\setminus \sigma_{2n}$ for all $n\ge n_0.$ For $T\in \{I,II,III\}$, define
\begin{equation*}
\Sigma_T:=\{E\in \R: E \text{ is a type-T energy}\}.
\end{equation*}

It is seen that $\Sigma_I$ is countable. Moreover, it is  shown in \cite{LQ} that $\Sigma_I$ is dense in $\sigma(H_\lambda)$.

\begin{prop}\label{extended}
For any $E\in \Sigma_I$, there exists a constant $C(E)>1$ such that for any formal solution $\psi$ of $H_\lambda\psi=E\psi$,
\begin{equation}\label{bd-transfer}
C^{-1}\|\vec{\psi}_0\|\le \|\vec{\psi}_n\|\le C\|\vec{\psi}_0\|\ \ \ \text{ and }\ \ \ 1\le \|T_n(E)\|\le C.
\end{equation}
\end{prop}

Compare with \eqref{transfer-periodic} and \eqref{solution-periodic}, we conclude that  $E$ is an extended state for any $E\in \Sigma_I.$ Proposition \ref{extended} follows quite easily from an observation stated in \cite{DT}, and implicitly proved in \cite{AP}.

Now we study the other two subsets.
By  \eqref{structure-spectrum} and the definitions of $\Sigma_{II}$ and $\Sigma_{III}$,  we know that
$
\Sigma_{II}, \Sigma_{III} \subset \sigma(H_\lambda).
$
Of course, apriori $\Sigma_{II}$ and $\Sigma_{III}$ may be empty.
However we have the  following:

\begin{thm}\label{exist-dense}
Both $\Sigma_{II} $ and $\Sigma_{III}$ are  dense in $\sigma(H_\lambda)$ and uncountable.

\end{thm}

 As we will see, the proof of the theorem is far from trivial, which involves the detailed analysis of the dynamics of a two-dimensional polynomial map, constructed  through the recurrence relation of the trace polynomials.

\smallskip

Next,  we obtain the estimations for the norms of transfer matrices for type-II and type-III energies.

\begin{thm}\label{main-norm}
For any $E\in \Sigma_{II}\cup \Sigma_{III}$,
there exist  constant $\gamma=\gamma(E)>0$ and absolute constants $c_2>c_1>0$ such that,
for $n$  large enough,
\begin{equation}\label{main-norm-eq}
e^{c_1\gamma\sqrt{n}}\le\|T_{ \pm n}(E)\|\le  e^{c_2\gamma\sqrt{n}}.
\end{equation}
 \end{thm}

\begin{remark}
{\rm
(i)\  \eqref{main-norm-eq} implies that $\|T_{\pm n}(E)\|$ increase at least with the speed $e^{c_1\gamma\sqrt{n}}$ for $n$ large enough, thus the exponent $\alpha=1/2$ in \eqref{uniform} is sharp.

(ii)\
One may expect that  the following limit exists
$$
\lim_{n\to\infty}\frac{\log \|T_{\pm n}(E)\|}{\sqrt{n}}.
$$
However, it is not the case,
as suggested by Remark \ref{osci-II} and Remark \ref{osci-III}.

(iii) It is clear that \eqref{main-norm-eq} is different from all of \eqref{transfer-periodic}, \eqref{transfer-random} and \eqref{power-Fibonacci}. To our best knowledge, this is the first time that  energies  with \eqref{main-norm-eq} are  presented.

}
\end{remark}

 Even the definitions of $\Sigma_{II}$ and $\Sigma_{III}$ are quite similar,
 the related spectral properties  are quite different. To state the result, we introduce a new  notation.
 Write ${v}_\theta:=(\cos\theta,-\sin\theta)^t$ for any $\theta\in \R.$

\smallskip

For type-II energy, we have the following result:

\begin{thm}\label{main-type-II}
For each $E\in\Sigma_{II}$, the equation $H_\lambda\psi=E\psi$ has a unique subordinate solution $\psi$ satisfying the following:

(i) There exist $\alpha>0$ and  $c\ne0$ such that $\vec{\psi}_0=v_{\frac{\eta\pi}{4}}$ with $\eta\in\{+,-\}$ and
\begin{equation}\label{subord-II}
\begin{cases}
\|\vec{\psi}_{\pm n}\|\lesssim n^\alpha, \ \ \|\vec{\psi}_{\pm 2^{2n}}\| \lesssim e^{-2^n\gamma},\\
\vec{\psi}_{2^{2n+1}}=-\eta\vec{\psi}_{-2^{2n+1}}=\pm c \ v_{-\frac{\eta\pi}{4}}+O(e^{-2^{n+1}\gamma}).
\end{cases}
\end{equation}

(ii) If $\phi$ is such that $H_\lambda\phi=E\phi$ and $\phi$ is independent with $\psi$, then
\begin{equation}\label{non-subord-II}
\|\vec{\phi}_{\pm n}\|\gtrsim e^{\frac{\gamma}{4}\sqrt{n}}.
\end{equation}
\end{thm}

\begin{remark}
{\rm
(i)\
\eqref{subord-II} and \eqref{non-subord-II} obviously implies  that
$\psi$ is the unique  subordinate solution of the related eigen-equation.
\eqref{subord-II} also reveals  another interesting phenomenon:
while the solution $\psi$ is sub-exponentially localized along the subsequences $\{ 2^{2n}: n\ge 1\}$ and $\{ -2^{2n}: n\ge 1\}$,
it is extended along the subsequences $\{ 2^{2n+1}: n\ge 1\}$
and $\{ -2^{2n+1}: n\ge 1\}$. So $\psi$ is an oscillated solution with polynomial bound.

(ii)\ Compare with \eqref{exp-decay-random} and \eqref{exp-increase-random}, it is reasonable to call
each $E\in \Sigma_{II}$ a {\it pseudo-localized} state  of $H_\lambda$,
and call $\psi$ the  {\it pseudo-eigenfunction} related to $E$.
Since $\Sigma_{II}$ is  dense in the spectrum,  we may  say that $H_\lambda$ exhibits {\it  pseudo-localization}.
}
\end{remark}

\smallskip

For type-III energy, we have the following result:

\begin{thm}\label{main-type-III}
For each $E\in\Sigma_{III}$, the equation $H_\lambda\psi=E\psi$ has a unique right-subordinate
solution $\psi^r$ and a unique left-subordinate solution $\psi^l$ satisfying the following:

(i) there exist  $\theta\ne \pm\frac{\pi}{4}+k\pi$ and  $\alpha>0$ such that $\vec{\psi}_0^r=v_{\theta}$, and
\begin{equation}\label{subord-r-III}
\begin{cases}
\|\vec{\psi}_{ n}^r\|\lesssim n^\alpha,&
\|\vec{\psi}_{- n}^r\|\gtrsim  e^{\frac{\gamma}{4}\sqrt{n}},\\
\|\vec{\psi}_{ 2^{2n+1}}^r\| \lesssim e^{-2^n\gamma},&
\vec{\psi}_{ 2^{2n}}^r=\pm \frac{1}{\sqrt{2}}v_{-\frac{\pi}{2}-\theta }+O(e^{-2^n\gamma});
\end{cases}
\end{equation}
$\vec{\psi}_0^l=v_{-\frac{\pi}{2}-\theta}$, and
\begin{equation}\label{subord-l-III}
\begin{cases}
\|\vec{\psi}_{-n}^l\|\lesssim n^\alpha,&
\|\vec{\psi}_{n}^l\|\gtrsim  e^{\frac{\gamma}{4}\sqrt{n}},\\
\|\vec{\psi}_{ -2^{2n+1}}^l\| \lesssim e^{-2^n\gamma},&
\vec{\psi}_{ -2^{2n}}^l=\pm \frac{1}{\sqrt{2}}v_{\theta }+O(e^{-2^n\gamma}).
\end{cases}
\end{equation}

(ii) If $\phi$ is such that $H_\lambda\phi=E\phi$ and $\phi\not\in {\rm span}(\psi^l)\cup {\rm span}(\psi^r)$, then
\begin{equation}\label{non-subord-III}
\|\vec{\phi}_{\pm n}\|\gtrsim  e^{\frac{\gamma}{4}\sqrt{n}}.
\end{equation}
\end{thm}

\begin{remark}
{\rm
(i) Indeed, $\psi^l$ is just a reflection of $\psi^r$ with respect to $1/2.$
Since $\theta\ne \pm\frac{\pi}{4}+k\pi$, we have $\vec{\psi}_0^r \nparallel \vec{\psi}_0^l$,
consequently $\psi^r$ and $\psi^l$ are independent and $\{\psi^r,\psi^l\}$ is  a basis of the space of the solutions of $H_\lambda\psi=E\psi.$

(ii)\
\eqref{subord-r-III},  \eqref{subord-l-III} and \eqref{non-subord-III} implies that
$\psi^r$ is the unique right-subordinate solution and $\psi^l$ is the unique left-subordinate solution.
However, in this case, there is  no nonzero two-sided subordinate solution.
Indeed if $\phi$ is a global subordinate solution, then $\phi$ is both a right- and left-subordinate solution, then
$
\phi=t_1\psi^r=t_2\psi^l,
$ for some $t_1,t_2\in \R.$ Since $\psi^l$ and $\psi^r$ are independent,
the only possibility is $t_1=t_2=0.$

(iii)\ Not like the type-II energy case, \eqref{subord-r-III} reveals  a more complex  phenomenon:
while the solution $\psi^r$ is sub-exponentially localized along the subsequence $\{ 2^{2n+1}: n\ge 1\}$,
it is extended along the subsequence $\{ 2^{2n}: n\ge 1\}$ and  sub-exponentially exploring  along $\Z^-$.
However, $\psi^r$ is still an oscillated solution with polynomial bound on $\Z^+$. For $\psi^l$ we have similar observation.

(iv)\ Compare with the classical definition of  localization, it is still reasonable to call
each $E\in \Sigma_{III}$ a {\it one-sided pseudo-localized} state of $H_\lambda$,
and call $\psi^r (\psi^l)$ a {\it right (left) pseudo-eigenfunction} related to $E$.
Since $\Sigma_{III}$ is  dense in the spectrum,  we may  say that $H_\lambda$ exhibits  {\it one-sided pseudo-localization}.
}
\end{remark}

We summarize the spectral properties of  all the   operators mentioned above in the following table, where $\Theta(\gamma)  $ means  $c\gamma\le  \Theta(\gamma) \le d \gamma$ for constants $c,d>0.$

\bigskip

\begin{tabular}{|c|c|c|c|}\hline Model \& Energy &$\|T_{\pm n}(E)\|$&subordinate $\|\vec{\psi}_{\pm n}\|$&other sol. $\|\vec{\phi}_{\pm n}\|$\\
\hline Periodic\ $\sigma(H)$&$\le C$&no&$\le C$\\ 
\hline Anderson\ $\sigma_p(H)$ &$e^{(\gamma+o(1))n}$& $e^{-(\gamma+o(1))n}$&$e^{(\gamma+o(1)) n}$\\
\hline Fibonacci\ $\sigma(H)$&$\le  n^\beta$&$\le n^\beta$&$\le n^\beta$\\ \hline Thue-Morse\  $\Sigma_{I}$&$\le C$&no&$\le C$\\ \hline Thue-Morse\  $\Sigma_{II}$&$e^{\Theta(\gamma)\sqrt{n}}$&$\le n^\alpha$&$e^{\Theta(\gamma)\sqrt{n}}$\\ \hline Thue-Morse\  $\Sigma_{III}$&$e^{\Theta(\gamma)\sqrt{n}}$&no&$e^{\Theta(\gamma)\sqrt{n}}$\\ \hline
\end{tabular}

\bigskip


With these properties in hand, by using  the subordinacy theory developed in \cite{JL1,JL00},  we can  estimate  the local dimensions of the spectral measure at these energies.

We introduce several more notations.
Given a finite measure $\mu$ on $\R$. Fix $x\in {\rm supp}\mu$ and $\alpha>0$. The  $\alpha$-density of $\mu$ at $x$ is defined as
\begin{equation*}
D_\mu^\alpha(x):=\limsup_{\epsilon\to0}\frac{\mu((x-\epsilon,x+\epsilon))}{(2\epsilon)^\alpha}.
\end{equation*}
The local Hausdorff dimension of $\mu$ at $x$ is defined as
\begin{equation*}
d_\mu(x):=\liminf_{\epsilon\to 0} \frac{\log\mu((x-\epsilon,x+\epsilon)) }{\log \epsilon}.
\end{equation*}
$d_\mu(x)$ gives the local H\"older exponent of the distribution of $\mu$ at $x$.
The relation of these two quantities is
\begin{equation}\label{cha-locdim}
d_\mu(x)=\sup\{\alpha:D_\mu^\alpha(x)=0 \}=\inf \{\alpha:D_\mu^\alpha(x)=\infty \}.
\end{equation}

Now assume $\mu$ is the spectral measure of $H_\lambda$,
we study the local Hausdorff dimension of $\mu$ at energy $E$ of type-I, II and III.
We divide the set $\Sigma_{III}$ further  into two parts:
\begin{equation*}
\Sigma_{III,1}:=\{E\in \Sigma_{III}: \vec{\psi}_0^r=v_0 \text{ or } v_{\frac{\pi}{2}}\};\ \ \
\Sigma_{III,2}=\Sigma_{III}\setminus \Sigma_{III,1}.
\end{equation*}

\begin{thm}\label{main-locdim}
We have
\begin{equation}\label{loc-dim}
\begin{cases}
d_\mu(E)\ge 1& E\in \Sigma_I,\\
d_\mu(E)=0& E\in \Sigma_{II},\\
d_\mu(E)\ge 2& E\in \Sigma_{III,1},\\
d_\mu(E)\ge 1& E\in \Sigma_{III,2}.
\end{cases}
\end{equation}
Consequently we have
\begin{equation*}
\mu(\Sigma_{I})=0;\ \ \ \dim_H \Sigma_{II}=0\ \ \ \text{ and }\ \ \ \mu(\Sigma_{III,1})=0.
\end{equation*}
\end{thm}

\begin{remark}
{\rm
(i)\ Loosely  speaking,  $\mu|_{\Sigma_I}$ is reminiscent of the spectral measure of periodic Hamiltonian, while $\mu|_{\Sigma_{II}}$ is quite close to  the spectral measure of almost Mathieu operator for $|\lambda|>2$, as studied in \cite{JL00}. The behavior of $\mu|_{\Sigma_{III}}$ is a bit strange, because if one consider the half-line operator, the spectral measure should be as singular as $\mu|_{\Sigma_{II}}$, however for the whole line operator, the spectral measure is very smooth in the sense that, the local H\"older exponent is very large(can be $\ge 2$). It is the same phenomenon as explained in \cite{JL00}(see the remark after Lemma 4 \cite{JL00}).

(ii)\   Unfortunately, we do not know which subset essentially supports the spectral measure.  It is left as an interesting and challenging  problem.

(iii)\ From fractal geometry point of view,  \eqref{loc-dim} suggests that $\mu$ is a multifractal measure, another interesting problem is   to estimate  the local dimensions for all the other  energies in the spectrum.

(iv)  $\Sigma_{III,1}$ contains at most two energies. Indeed,  $\Sigma_{III,1}$ is empty for $|\lambda|>1$; on the other hand, there exists a uncountable subset $\Gamma\subset [-1,1]$, such that $\Sigma_{III,1}=\{\lambda,-\lambda\}$ for any $\lambda\in \Gamma.$ See Lemma \ref{sigma-3-1} for detail.  For the coupling parameter $\lambda\in \Gamma$, the energies $\pm \lambda$ may have some physical implication due to their strange local dimensions.
}
\end{remark}

The following table offers a summary  on  these three types of  energies.

\bigskip

\begin{tabular}{|c|c|c|c|c|c|c|}
\hline
\multirow{2}{*}{ Type} & \multicolumn{2}{c|}{\multirow{2}{*}{Type I}}& \multicolumn{2}{c|}{\multirow{2}{*}{Type II}}&
\multicolumn{2}{c|}{Type III}
 \\
\cline{6-7}&\multicolumn{2}{c|}{} &\multicolumn{2}{c|}{}
  & III-1 & III-2   \\
\hline
Existence &\multicolumn{2}{c|}{{Yes}}&\multicolumn{2}{c|}{{Yes}}& Sometimes Yes & Yes \\
\hline
Ordinality &\multicolumn{2}{c|}{{Countable}}&\multicolumn{2}{c|}{{Uncountable}}& $0$ or $2$ & Uncountable \\
\hline
Denseness &\multicolumn{2}{c|}{{Yes}}&\multicolumn{2}{c|}{{Yes}}& No & Yes \\
\hline
 $d_{\mu}(\cdot)$&\multicolumn{2}{c|}{{$\geq 1$}}&\multicolumn{2}{c|}{{0}}& $\geq 2$   & $\geq 1 $ \\
 \hline
\end{tabular}


\subsection{Ideas of the proof}

In the following, we explain several key points in our proof.

At first we  explain the proof of the existence and  denseness of $\Sigma_{II}$ and $\Sigma_{III}$.  We take $\Sigma_{III} $ as example. Based on the recurrence relation of trace polynomials(see \eqref{recurrence}),
we define a mapping $f:\mathbb{R}^2\to \mathbb{R}^2$ as
\begin{equation*}\label{def-2-dim-map}
f(x,y)=(x^2(y-2)+2,x^2y^2(y-2)+2).
\end{equation*}
Then   for any $E\in\mathbb{R}$ and any $n\ge1$,
\begin{equation*}
f(t_n(E),t_{n+1}(E))=(t_{n+2}(E),t_{n+3}(E)).
\end{equation*}

Write  $D:=[-2,2]\times(-\infty,-2).$
In view of the definition of   type-III energies and a simple fact deduced from the recurrence relation(see Lemma \ref{basic} (iii)),  we have
\begin{equation}\label{equi-def}
\Sigma_{III}=\{E\in \R: f^n(t_{2k-1}(E),t_{2k}(E))\in D, \ \  \exists k\in \N, \forall n \ge 0\},
\end{equation}
 If we write
$$
\varphi_k(E):=(t_{k}(E),t_{k+1}(E)) \ \ \text{ and }\  \ \Lambda(D):=\bigcap_{n\ge 0} f^{-n}(D),
$$
we can rewrite \eqref{equi-def} as
\begin{eqnarray*}
\Sigma_{III} =\bigcup_{k \text{ odd}}\{E:  \varphi_k(E)\in \Lambda(D)\}.
\end{eqnarray*}
Thus $\Sigma_{III}$ is nonempty  if  the curve $\varphi_k(\R)$ intersects $\Lambda(D)$ for some odd $k$. It is dense in the spectrum if for any $E\in \sigma(H_\lambda)$ and $\epsilon>0$, the sub-curve $\varphi_k([E-\epsilon,E+\epsilon])$ intersects $\Lambda(D)$ for some odd $k$.  Hence  it is crucial to understand $\Lambda(D).$
However, it turns out that  this set is very hard to deal with, because $D$ is not matched with the dynamical properties of $f$. On the other hand, observe  that the parabola $y=x^2-2$ is invariant under $f$ and $(-1,-1)$ is a fixed point of $f$, then a  more suitable starting domain is
\begin{equation*}\label{def-S}
S=\{(x,y):-1\le x\le 1;\ y\le x^2-2\}.
\end{equation*}


\begin{figure}[h]
\includegraphics[width=0.8\textwidth]{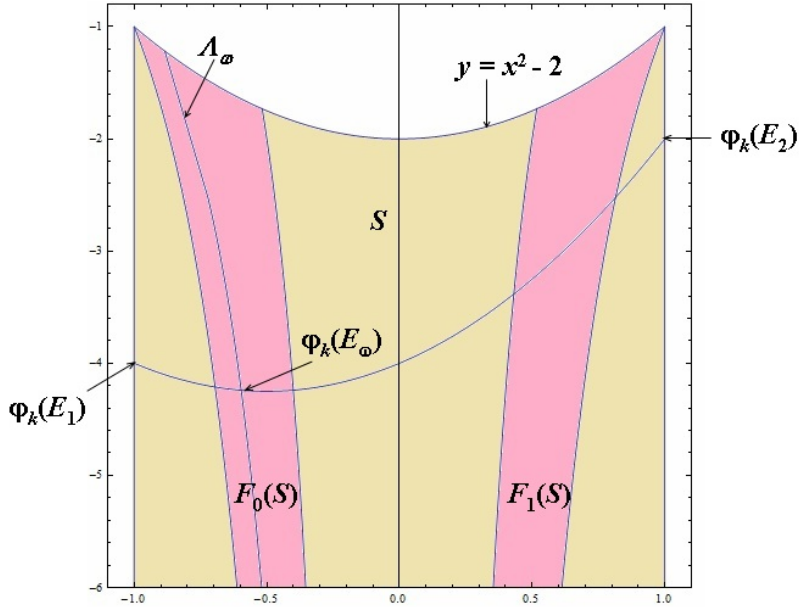}
\caption{Proof of the denseness of $\Sigma_{III}$ }\label{pic-sigma-3}
\end{figure}


Indeed, by a bit extra work, we can show that
\begin{eqnarray*}
\Sigma_{III}=\bigcup_{k \text{ odd}} \{E: \varphi_k(E)\in \Lambda(S)\}.
\end{eqnarray*}
The good point is that the structure of $\Lambda(S)$ is relatively simple. In fact,
$f$ has four inverse branches $F_0,F_1, F_2,F_3$ with
$$
F_0(S), F_1(S)\subset S, \ \ F_0(S)\cap F_1(S)=\emptyset\ \ \text{ and }\ \  F_2(S)\cap S=F_3(S)\cap S=\emptyset.
$$
Thus
$$
\Lambda(S)=\bigcap_{n\ge 0}f^{-n}(S)=\bigcap_{n\ge 1}\bigcup_{w_i\in \{0,1\}} F_{w_1}\circ \cdots\circ F_{w_n}(S).
$$
$\Lambda(S)$ is kind of  repeller and
has a ``lamination" structure.  Moreover, each ``fiber" $\Lambda_\omega:=\bigcap_{n}F_{\omega_1}\circ \cdots\circ  F_{\omega_n}(S)$ is rooted in the parabola $y=x^2-2$ and goes down  to $-\infty$.  Now, due to the denseness of $\Sigma_{I}$, for each $E\in \sigma(H_\lambda)$ and $\epsilon>0$, we can find  odd $k$ large enough and a subinterval $[E_1,E_2]\subset [E-\epsilon,E+\epsilon]$,  such that the sub-curve
$\varphi_k([E_1,E_2])$ stays in $S$ and touches the left and right boundaries of  $S$.  As a result, this sub-curve must intersect with each fiber $\Lambda_\omega$, which finishes the proof(see Figure \ref{pic-sigma-3}).

Next, we explain  how to get the lower bound of the norms of the transfer matrices for type-II and type-III energies. We take type-II energy as example. Fix $E\in \Sigma_{II}$,  due to the combinatorial nature of Thue-Morse potential, it is natural to consider firstly the norms of
$$
A_n=T_{2^n} \ \ \ \text{ and } B_n=T_{2^n\to 2^{n+1}}.
$$
Define four special matrices as follows:
\begin{equation}\label{4matrices}
I=
\begin{pmatrix}
1& 0\\
0& 1
\end{pmatrix}
, \
U=
\begin{pmatrix}
0& 1\\
1& 0
\end{pmatrix}
, \
V=\begin{pmatrix}
1& 0\\
0& -1
\end{pmatrix}
, \
W=\begin{pmatrix}
0& -1\\
1& 0
\end{pmatrix}
.\end{equation}
  A key observation  in  \cite{DP} is that, 
\begin{equation}\label{key-1}
A_{2n-1}-B_{2n-1}=\mu_n U
\end{equation}
for some $\mu_n\in \R$.
By this observation, we obtain 
\begin{equation}\label{key-2}
A_{2n}-B_{2n}=\nu_n V+\omega_n W
\end{equation}
for some $\nu_n,\omega_n\in \R.$
 Combine \eqref{key-1} with  the recurrence relations of $A_n$ and $B_n$, we get
$$
A_{2n+1}=t_{2n}t_{2n-1}A_{2n-1}-t_{2n-1}\mu_{n}U+(1-t_{2n-1}^2)I,
$$
where $t_n$ is the $n$-th trace polynomial.
By studying the recurrence relation of $t_n$(see \eqref{t12} and \eqref{recurrence}), we show that  there is $\gamma>0$ such that
$$
|t_{2n-1}|\sim e^{2^n\gamma},\quad |t_{2n}|\sim e^{-2^n\gamma}\quad \text{ and } \quad |\mu_n|\sim e^{2^n\gamma}.
$$ By more precise estimations on these three quantities, we can show that $\{t_{2n-1}^{-1}A_{2n-1}:n\ge 1\}$ is a Cauchy sequence and moreover
\begin{equation}\label{keyvec}
\frac{A_{2n-1}}{t_{2n-1}}= \frac{I\pm U}{2}+O(e^{-2^{n-1}\gamma}).
\end{equation}
Similarly but more complicatedly, by using \eqref{key-2}, we can show that 
\begin{equation}\label{keyvec-even}
t_{2n}A_{2n}=\pm c(V\mp W)+O(e^{-2^{n+1}\gamma}),
\end{equation}
 where $c$ is a nonzero constant. These two equations suggest that
$$
\|A_{2n}\|, \|A_{2n-1}\| \sim e^{2^n\gamma}.
$$
Now an interpolation argument shows that for any $n\in \N$, we have $\|T_n\|\gtrsim  e^{c\gamma \sqrt{n}}$.  Figure \ref{pic-log-norm} is a simulation on the norms of the transfer matrices for some $E\in \Sigma_{II}$, which matches very well with the theoretical prediction \eqref{main-norm-eq}.


\begin{figure}[h]
\includegraphics[width=0.7\textwidth]{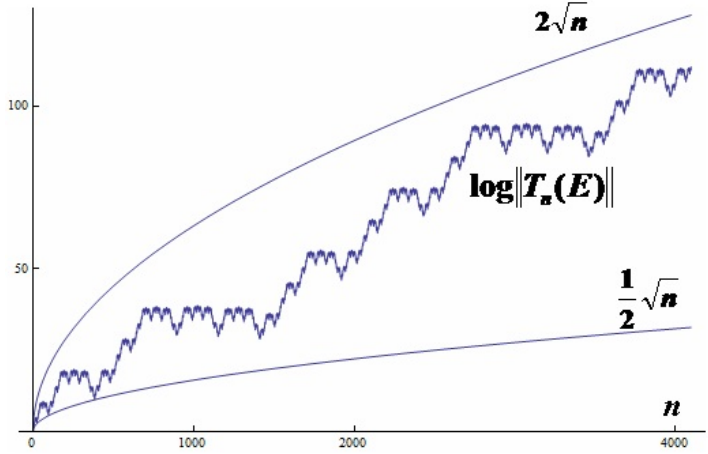}
\caption{$\log\|T_n(E)\|, n=1,\cdots, 4096$, where $\lambda=1, E = 1.571613\cdots$}\label{pic-log-norm}
\end{figure}


Now we explain how to find the   subordinate solution at $\Z^+$. By \eqref{transfer-vector}, we only need to find the suitable initial condition $\vec{\psi}_0$ such that  the norm of $\vec{\psi}_n=T_n\vec{\psi}_0$ is as small as possible. In the language of dynamical system, we should find the stable direction of $\{T_n:n\in\N\}$.

If the constants $c_1$ and $c_2$ in Theorem \ref{main-norm} satisfy $c_2<2c_1$, then we could apply Oseledets' argument
to find a stable direction $s$ such that $\|T_ns\|\lesssim e^{-c\gamma\sqrt{n}}$,
and for any vector $v$ that is linear independent with $s$, $\|T_nv\|\gtrsim e^{c\gamma\sqrt{n}}$. However,  comparing with Theorem \ref{main-type-II} and Theorem \ref{main-type-III}, it is impossible to choose $c_1, c_2$ such that $c_2<2c_1$ in our case. On the other hand,
  one can always apply Oseledets' argument to find a ``pre"-stable direction $s$ such that
for any vector $v$ that is linear independent with $s$, 
\begin{equation}\label{est-other}
\|T_nv\|\gtrsim e^{c\gamma\sqrt{n}}.
\end{equation}
 This $s$ should be the initial condition $\vec{\psi}_0$ we are looking for,
but this approach say nothing about $\|T_ns\|$ itself. 

Instead of considering the whole sequence $\{T_n:n\in\N\}$, we treat the subsequence $\{A_n=T_{2^n}: n\in\N\}$ firstly. The advantage is that, in this case we know the exact asymptotic behaviors of  $\|A_n\|$ due to \eqref{keyvec} and \eqref{keyvec-even}.  But here there is a further subtlety.   It turns out there is no stable direction for $\{A_n: n\in \N\}$, and we need to consider a further subsequence. It is at this point that  we need to treat type-II
and type-III energies separately. 

At first suppose $E\in\Sigma_{II}$.
By Oseledets' argument, there exists a  stable direction $\tilde s$ for $\{A_{2n}: n\in \N\}$ such that
$\|A_{2n}\tilde s\|\lesssim e^{- 2^n\gamma}.$ In view of \eqref{est-other}, $s=\tilde s.$
Moreover,  by \eqref{keyvec-even},  $s$ must be parallel to   $v_{\pm\pi/4}$.
By some extra work, we can show that  $\|A_{2n+1}s\|\sim 1$ and $\|T_ns\|\le n^\alpha$ for  some $\alpha>0$ and any $n$. As a result,
we find the  subordinate solution at $\Z^+$.

Next suppose $E\in\Sigma_{III}$.
Again by Oseledets' argument, 
there exists a  stable direction $\hat s$ for $\{A_{2n+1}: n\in \N\}$ such that
$\|A_{2n+1}\hat s\|\lesssim e^{- 2^n\gamma}.$ Still by \eqref{est-other}, $s=\hat s.$ However,   $s$ is never parallel to $v_{\pm\pi/4}$ in this case. Indeed, we can show a formula analogous to \eqref{keyvec}:
$$\frac{A_{2n}}{t_{2n}}= \frac{1}{2}(I\pm\sec\theta\  V\pm\tan\theta\  W) +O(e^{-2^{n-1}\gamma}),$$
where $\theta\in(-\pi/2,\pi/2)$ is determined by $\sin\theta=(E^2-\lambda^2)/2E$.
By a simple computation, we get  $s\parallel v_{-\theta/2}$ or $v_{(\theta-\pi)/2}$. Thus $s\nparallel v_{\pm\pi/4}$.
We can still show that  $\|A_{2n}s\|\sim 1$ and $\|T_ns\|\le n^\alpha$ for  some $\alpha>0$ and any $n$, so
we  also find a  subordinate solution at $\Z^+$ for type-III energy.

Let us check whether this one-sided subordinate solution extends to a global one. From the definition, it is seen that the Thue-Morse potential is symmetric with respect to $1/2$. If $E$ is a type-II energy, then $\vec{\psi}_0=(\psi_1,\psi_0)=s\parallel v_{\pm\pi/4}$, which is either symmetric or anti-symmetric with respect to $1/2$. Consequently, $\psi$ extends to a global subordinate solution. As a comparison, if $E$ is a type-III energy, then $\vec{\psi}_0\nparallel v_{\pm\pi/4}$, hence $\psi$ never extends to a global subordinate solution. 
 
 \bigskip

 Finally, we say some words on notations.

By $a_n\lesssim b_n$, we mean
there exists $C>0$ such that, for any $n>0$, $a_n\le C b_n.$
The notation $\gtrsim$ can be defined analogously.
By $a_n\sim b_n$, we mean $a_n\lesssim b_n$ and $a_n\gtrsim b_n$.

We denote by $\|x\|$ the $2$-norm of $x\in\mathbb{R}^2$.
For any matrix $A$ of order $2$,
$$\|A\|:=\sup_{\|x\|=1}\|Ax\|.$$
Note that if $\det A=1$, then $\|A\|=\|A^{-1}\|\ge1$.

The rest of the paper is organized as follows. In Section \ref{sec-preliminary}, we give basic definitions and some useful facts related to  Thue-Morse potential and the first structure of the transfer matrices. In Section \ref{sec-th1-prop2}, we prove Theorem \ref{main-subexp} and Proposition \ref{extended}, which are relatively easy. In Section \ref{sec-exist-dense}, we prove Theorem \ref{exist-dense}. In Section \ref{sec-type-II}, we prove Theorem \ref{main-norm} for type-II energy and Theorem \ref{main-type-II}. In Section \ref{sec-type-III}, we prove Theorem \ref{main-norm} for type-III energy and Theorem \ref{main-type-III}. In Section \ref{sec-loc-dim}, we prove Theorem \ref{main-locdim}. Finally, in the Appendix, we give an alternative proof of the fact that $H_\lambda$ has no point spectrum, which is based on our new understanding on the operator.

\section{Preliminaries}\label{sec-preliminary}

In this section, we give the basic definitions and collect some basic facts related to trace polynomials. We also give the first structure of the transfer matrices.
\subsection{Basic definitions}\

Define  $\bar{a}:=b$ and $\bar{b}:=a.$ For $a_1\cdots a_n\in \{a,b\}^n$, define $\overline{a_1\cdots a_n}:=\overline{a}_1\cdots \overline{a}_n.$
Recall that $\varsigma^\infty(a)=u=u_1u_2\cdots$ is  the Thue-Morse sequence.
 It has the following symmetric property:
\begin{equation}\label{parlin}
\overline{u_{2^{n}+1}\cdots u_{2^{n+1}}}=u_1\cdots u_{2^{n}};\ \ u_1\cdots u_{2^{2n}}= u_{2^{2n}}\cdots u_1,\quad \forall n>0.
\end{equation}
We extend $u$  to a  two-sided Thue-Morse sequence as follows
\begin{equation}\label{symmetry-TM}
u_{1-n}=u_n,\quad \forall n>0.
\end{equation}
By this definition, $u$ is symmetric with respect to $1/2$.

 Denote the free group generated by $a,b$ as ${\rm FG}(a,b).$ Given $\lambda, E\in \R$ and $\lambda\ne 0,$  define a homomorphism $\varrho:{\rm
FG}(a,b)\to {\rm SL}(2,\R)$ as
$$\varrho(a)= \left[
\begin{array}{cc}
E-\lambda&-1\\
1&0
\end{array}
\right],\  \ \ \ \varrho(b)= \left[
\begin{array}{cc}
E+\lambda&-1\\
1&0
\end{array}
\right]
$$
and $\varrho(a_1\cdots a_n)=\varrho(a_n)\cdots\varrho(a_1).$ The transfer matrix $T_{m\to n}(E)$ can be rewritten   as
$$
T_{m\to n}(E)=
\begin{cases}
\varrho(u_{m+1}\cdots u_n)& m< n\\
I&m=n\\
\left(\varrho(u_{n+1}\cdots u_m)\right)^{-1}& m>n
\end{cases}
$$
It is basic that $T_{m\to n}(E)$ has determinant 1.

Define $\bar{\varrho}(x)=\varrho(\bar{x})$ for $x=a$ and $b$ and extend the definition of $\bar{\varrho}$ to finite word.
Define
$$
\overline{T}_{m\to n}(E):=
\begin{cases}
\bar{\varrho}(u_{m+1}\cdots u_n)& m< n\\
I&m=n\\
\left(\bar{\varrho}(u_{n+1}\cdots u_m)\right)^{-1}& m>n
\end{cases}
$$

For $n\ge 0$, define
\begin{equation*}
A_n(E):=\varrho(\varsigma^n(a))=T_{2^n}(E)\ \ \ \text{ and }\ \ \  B_n(E):=\varrho(\varsigma^n(b))=\overline{T}_{2^n}(E).
\end{equation*}
Then for $n\ge 0,$ we have
\begin{equation}\label{recu-n+1-n}
A_{n+1}(E)=B_n(E)A_n(E),\quad B_{n+1}(E)=A_n(E)B_n(E).
\end{equation}
As a consequence, ${\rm tr}(A_n(E))={\rm tr}(B_n(E))$ for $n\ge 1$(where ${\rm tr}(A)$ denotes the trace of the matrix $A$).
Recall that we have defined
$
t_n(E)={\rm tr}(T_{2^n}(E))={\rm tr}(A_n(E)).
$
It is shown in \cite{AP,B} that
\begin{equation}\label{t12}
t_1(E)=E^2-\lambda^2-2,\quad t_2(E)=(E^2-\lambda^2)^2-4E^2+2,
\end{equation}
and for $n\ge2$,
\begin{equation}\label{recurrence}
t_{n+1}(E)=t_{n-1}^2(E)(t_n(E)-2)+2.
\end{equation}
$\{t_n: n\ge 1\}$ is the family  of { trace polynomials} related to the Thue-Morse Hamiltonian.

The follow lemma collects  basic facts about  the trace polynomials.

\begin{lemma}\label{basic}
(i) If $t_n(E)=0$, then $t_{k}(E)=2$ for any $k\ge n+2.$ Consequently $\Sigma_I\cap (\Sigma_{II}\cup \Sigma_{III})=\emptyset.$

(ii) If $n\ge 2$ and $t_n(E)=2$, then $t_{k}(E)=2$ for any $k\ge n.$

(iii) If $E\in\sigma(H_\lambda),$ $t_{n-1}(E)\ne0$ and $|t_n(E)|>2$ for some $n>1$, then $t_n(E)<-2$.

(iv) $0\not \in \sigma(H_\lambda)$.
\end{lemma}
\begin{proof}
 (i)  follows directly from \eqref{recurrence} and the definitions of $\Sigma_I, \Sigma_{II}$ and $\Sigma_{III}$.

 (ii) follows directly from \eqref{recurrence}.

By \eqref{recurrence}, if $t_n(E)>2$, then $t_{n+1}(E)=t_{n-1}^2(E)(t_n(E)-2)+2>2$, which contradicts with \eqref{structure-spectrum}. Thus (iii) holds.

By \eqref{t12}, we have $t_1(0)=-\lambda^2-2<-2$ and $t_2(0)=2+\lambda^4>2$.
Thus $0\not\in \sigma_1\cup \sigma_2. $ By \eqref{structure-spectrum}, $0\not\in \sigma(H_\lambda).$
\end{proof}

We usually drop $E$ from the notations
when there is no confusion.
For example, we simplify $A_n(E)$, $T_n(E)$, $t_n(E)$
to $A_n$, $T_n$, $t_n$.

\subsection{Structures of $A_n$ and $B_n$}\

It is a simple observation that, for any $2\times2$ matrix $A=(a_{ij})$ with determinant $1$, 
\begin{equation}\label{antisym}
UAU=A^{-1}\quad\mbox{if and only if}\quad a_{12}+a_{21}=0,
\end{equation}
where  $U$ is defined in \eqref{4matrices}.  As a consequence, we have
\begin{lemma}[\cite{DP,DGR}]\label{dp}
(i) $U\varrho(a)U=\varrho^{-1}(a)$ and $U\varrho(b)U=\varrho^{-1}(b)$.

(ii) More generally, for a word $a_1\cdots a_n\in\{a,b\}^n$,
\begin{equation}\label{inverse}
U\varrho(a_1\cdots a_n)U=\varrho^{-1}(a_n\cdots a_1).
\end{equation}
(iii) In particular, for any $n\ge 1$, we have
\begin{equation}\label{formu-reverse}
T_{-n}(E)=UT_n(E)U.
\end{equation}
\end{lemma}
\begin{proof}
 (i) and (ii) follow directly  from \eqref{antisym} and the definition of $\varrho$.

By (ii) and the definition of $T_n(E)$, 
$$
UT_n(E)U=U\varrho(u_1\cdots u_n)U=\varrho^{-1}(u_n\cdots u_1).
$$
On the other hand, by \eqref{symmetry-TM},
$$
T_{-n}(E)=(T_{-n\to 0}(E))^{-1}=\varrho^{-1}(u_{1-n}\cdots u_0)=\varrho^{-1}(u_n\cdots u_1).
$$
Thus (iii) holds.
\end{proof}

\begin{rem}\label{rembasic}
{\rm
By \eqref{parlin} and Lemma \ref{dp} (ii), for any $n>0$,
we have $UA_{2n}U=A_{2n}^{-1}$ and $UB_{2n}U=B_{2n}^{-1}$.
By \eqref{antisym}, for any $n\ge 0$, if we write
$$A_{2n}=\begin{pmatrix}a_{11}&a_{12}\\ a_{21}&a_{22}\end{pmatrix},\quad
B_{2n}=\begin{pmatrix}b_{11}&b_{12}\\ b_{21}&b_{22}\end{pmatrix},$$
then $a_{12}=-a_{21}$ and $b_{12}=-b_{21}$. See also \cite{DP,DGR}.

By \eqref{recu-n+1-n} and direct computation, for any $n>0$, if we write
$$A_{2n-1}=\begin{pmatrix}a_{11}'&a_{12}'\\ a_{21}'&a_{22}'\end{pmatrix},\quad
B_{2n-1}=\begin{pmatrix}b_{11}'&b_{12}'\\ b_{21}'&b_{22}'\end{pmatrix},$$
then
$a_{11}'=b_{11}',\ a_{12}'=-b_{21}',\ a_{21}'=-b_{12}',\ a_{22}'=b_{22}'$.
}
\end{rem}

\begin{lemma}\label{recurrence-AB}
We have
\begin{equation*}
\begin{array}{rcl}
A_{n+2}&=&t_n(t_{n+1}-1)A_n+t_nB_n+(1-t_n^2)I,\\
B_{n+2}&=&t_n(t_{n+1}-1)B_n+t_nA_n+(1-t_n^2)I,\\
B_n^2A_n&=&t_nA_{n+1}-A_n,\ \
A_n^2B_n=t_nB_{n+1}-B_n.
\end{array}
\end{equation*}
\end{lemma}
The proof follows from Cayley-Hamilton Theorem for  matrix $A$ of order $2$, i.e.
$A^2-(\tr A)A+(\det A)I=0.$

The following relation between $A_{2n-1}$ and $B_{2n-1}$ is observed in \cite{DP}:
\begin{lemma}\cite{DP}
There exists a real sequence $\{\mu_n\}_{n=1}^{\infty}$ such that
\begin{equation}\label{A-B-odd}
 A_{2n-1}-B_{2n-1}=\mu_{n}U \ \ \text{ and  }\ \  t_{2n}=t_{2n-1}^2-\mu_{n}^2-2.
\end{equation}
\end{lemma}

We observe the following parallel relation between  $A_{2n}$ and $B_{2n}:$

\begin{lemma}
There exist two real sequences  $\{\nu_{n}\}_{n=1}^{\infty}$ and $\{\omega_{n}\}_{n=1}^{\infty}$  such that
\begin{equation}\label{A-B-even}
A_{2n}-B_{2n}=\nu_{n}V+\omega_{n}W \ \ \text{ and }\ \  t_{2n+1}=t_{2n}^2-\nu_{n}^2+\omega_{n}^2-2,
\end{equation}
where $V$ and $W$ are defined in \eqref{4matrices}.
\end{lemma}

\begin{proof}
Since $A_{2n-1}U-UA_{2n-1}$ is a linear combination of $V$ and $W$, by \eqref{A-B-odd}, there exist $\nu_n,\omega_n\in\R$ such that
\begin{align*}
A_{2n}-B_{2n}&=B_{2n-1}A_{2n-1}-A_{2n-1}B_{2n-1}\\
&=(A_{2n-1}-\mu_{n}U)A_{2n-1}-A_{2n-1}(A_{2n-1}-\mu_{n}U)\\
&=\mu_{n}(A_{2n-1}U-UA_{2n-1})\\
&=\nu_{n}V+\omega_{n}W.
\end{align*}

Noticing  that $\det(\nu_{n}V+\omega_{n}W)=\omega_n^2-\nu_n^2$
and $\tr(\nu_{n}V+\omega_{n}W)=0$,
by Cayley-Hamilton Theorem,
$$\begin{array}{rcl}
-(\omega_n^2-\nu_n^2)I&=&(A_{2n}-B_{2n})^2\\
&=&A_{2n}^2+B_{2n}^2-A_{2n}B_{2n}-B_{2n}A_{2n}\\
&=&t_{2n}A_{2n}+t_{2n}B_{2n}-A_{2n+1}-B_{2n+1}-2I.
\end{array}
$$
Taking trace on both sides, we get the second equality of the lemma.
\end{proof}

Moreover, $\mu_{n}$, $\nu_{n}$ and $\omega_{n}$ satisfy the following recurrence relations:

\begin{lemma}\label{recur-mu-nu-omega}
The following three recurrence relations hold:
\begin{eqnarray*}
\mu_{n+1}&=&(t_{2n}-2)t_{2n-1}\ \mu_{n},\\
\nu_{n+1}&=&(t_{2n+1}-2)t_{2n}\ \nu_{n},\\
\omega_{n+1}&=&(t_{2n+1}-2)t_{2n}\ \omega_{n}
\end{eqnarray*}
with  initial conditions
\begin{equation}\label{initial-uvw}
\mu_1=-2\lambda;\ \  \nu_1=4\lambda E \ \ \text{ and }\ \  \omega_1=2\lambda(E^2-\lambda^2).
\end{equation}

\end{lemma}

\begin{proof}
At first, by direct computation, we get \eqref{initial-uvw}.

By Lemma \ref{recurrence-AB} and \eqref{A-B-odd} we have
\begin{equation}
\begin{cases}\label{recur-AB-odd}
A_{2n+1}&=t_{2n}t_{2n-1}A_{2n-1}-t_{2n-1}\mu_{n}U+(1-t_{2n-1}^2)I,\\
B_{2n+1}&=t_{2n}t_{2n-1}B_{2n-1}+t_{2n-1}\mu_{n}U+(1-t_{2n-1}^2)I.
\end{cases}
\end{equation}
Subtracting the two equations, we get the first recurrence relation.

By Lemma \ref{recurrence-AB} and \eqref{A-B-even} we have
\begin{equation}
\begin{cases}
\label{recur-AB-even}
A_{2n+2}&=t_{2n}t_{2n+1}A_{2n}-t_{2n}\nu_{n}V-t_{2n}\omega_nW+(1-t_{2n}^2)I,\\
 B_{2n+2}&=t_{2n}t_{2n+1}B_{2n}+t_{2n}\nu_{n}V+t_{2n}\omega_nW+(1-t_{2n}^2)I.
\end{cases}
\end{equation}
Subtracting the two equations, we get the second and the third recurrence relations.
\end{proof}

Finally we need the following basic fact to estimate  the norms  of $A_n$ and $B_n,$ the proof of which is elementary and will be omitted.

\begin{lemma}\label{elementary}
Assume $\gamma, c, d>0$ and   $\{a_n:n\ge1\}$ is a nonnegative sequence. If  either
$
a_{n+1}\le c+d \cdot e^{-2^n\gamma} a_n$ for all $n\in\N$ or $ a_{n+1}\le (1+c\cdot e^{-2^n\gamma}) a_n+d\cdot e^{-2^n\gamma}
$ for all $n\in\N,$
then  $\{a_n:n\ge1\}$ is bounded.
\end{lemma}

\section{Upper bounds for $\|T_{\pm n}(E)\|;$ type-I energy}\label{sec-th1-prop2}

In this section, we give the proofs of Theorem \ref{main-subexp} and Proposition \ref{extended}, which are relatively easy.

\subsection{Uniform upper bound for $\|T_{\pm n}(E)\|$}\

Note that for a matrix $A$ of order 2, we have $|{\rm tr}(A)|\le 2\|A\|.$

\noindent{\bf Proof of Theorem \ref{main-subexp}.}\
By Remark \ref{rembasic},  $\|A_{2n-1}\|=\|B_{2n-1}\|$ for any $n\ge1$.

\smallskip

\noindent{\bf Claim:} $\|A_{2n+1}\|\le 12\|A_{2n-1}\|^2.$

\noindent$\lhd$
Since  $t_n=\tr A_n=\tr B_n$ for any $n\ge1$, we have
\begin{equation}\label{trnorm}
|t_n|\le2\min\{ \|A_n\|,\ \|B_n\|\}.
\end{equation}

By Lemma \ref{recurrence-AB},
$$A_{2n+1}=t_{2n-1}(t_{2n}-1)A_{2n-1}+t_{2n-1}B_{2n-1}+(1-t_{2n-1}^2)I.$$

If $|t_{2n}|\le 2$,  then by using $|t_{2n-1}|\le2\|A_{2n-1}\|$, we have
$$\|A_{2n+1}\|\le 12\|A_{2n-1}\|^2.$$

If $|t_{2n}|>2$ and $t_{2n-1}=0$, then $A_{2n+1}=I.$ Thus the claim obviously holds since $\|A_{2n-1}\|\ge 1.$

If $|t_{2n}|>2$ and $t_{2n-1}\ne0$,  then  $t_{2n}<-2$ by Lemma \ref{basic} (iii).
Consequently  $|t_{2n+1}|, \ |t_{2n-1}|\le2$ by \eqref{structure-spectrum}. Hence  by \eqref{recurrence},
$$|t_{2n-1}|=\sqrt{\frac{2-t_{2n+1}}{2-t_{2n}}}\le\frac{2}{\sqrt{|t_{2n}|}}.$$
On the other hand, $|t_{2n}|\le 2\|A_{2n}\|\le 2\|A_{2n-1}\|\|B_{2n-1}\|=2\|A_{2n-1}\|^2,$ so
$$\begin{array}{rcl}
\|A_{2n+1}\|&\le&|t_{2n-1}t_{2n}|\|A_{2n-1}\|+2|t_{2n-1}|\|A_{2n-1}\|+|1-t_{2n-1}^2|\\
&\le& 2\sqrt{|t_{2n}|}\|A_{2n-1}\|+8\|A_{2n-1}\|
\le12 \|A_{2n-1}\|^2.
\end{array}
$$
Thus the claim holds.
\hfill $\rhd$

It is seen that $|E|\le |\lambda|+2$ for any $E\in\sigma(H_\lambda).$ Thus
$$
\|A_1(E)\|=\|\varrho(ab)\|\le\|\varrho(a)\|\|\varrho(b)\|\le 4(2+|\lambda|)^2.
$$
 Define $c_1=2(\ln2+ \ln(2+|\lambda|)$,
 define $c_n=c_{n-1}+2^{-n+1}\ln 12$ for any $n>1$,
it is direct to show by induction that
$$\|A_{2n-1}\|\le e^{c_n2^{n-1}}.$$
Then for any $n>0$,
$$\max\{\|A_{2n}\|,\|B_{2n}\|\}\le \|A_{2n-1}\|^2\le e^{c_n 2^{n}}.$$
Let $c=c_1+\ln 12$. For any $n>0$, we get
$$\max\{\|A_{n}\|,\|B_{n}\|\}\le e^{c2^{n/2}}.$$

For any $n>0$, let $k=\lfloor\log n/\log 2\rfloor$, then
$n=\sum_{j=0}^k a_j 2^j$ with $a_j\in\{0,1\}$ and $a_k=1$.
Define $n_0=0;$ for any $1\le i\le k+1$, define $n_i=\sum_{j=0}^{i-1} a_{k-j} 2^{k-j}$.
Then    $T_{n_{j}\rightarrow n_{j+1}}=A_{k-j}$ or $B_{k-j}$ if $a_{k-j}=1$.
And hence,
$$\|T_n\|\le\prod_{j=0}^{k}\|T_{n_{j}\rightarrow n_{j+1}}\|\le \prod_{j=0}^k e^{a_jc 2^{j/2}}
\le e^{\frac{\sqrt{2}c}{\sqrt{2}-1}(2^{k})^{1/2}}
\le e^{\frac{\sqrt{2}c}{\sqrt{2}-1}n^{1/2}}.$$

By Lemma \ref{dp} (iii) and the fact that $U$ is unitary, we have
  $\|T_{-n}\|=\|T_n\|$ for any $n\in\N$. Thus the result follows.
\hfill $\Box$

Compare with the upper bound of the transfer matrices, it is much more difficult to obtain a non-trivial lower bound, as we will see in the proof of Theorem \ref{main-norm}.

\subsection{Type-I energy}\

  The following simple fact is observed in \cite{AP,DT}, which also follows directly  from Lemma \ref{recurrence-AB}.

\begin{lemma}
If $t_k(E)=0$ for some $k$, then $A_{k+2}=B_{k+2}=I.$
\end{lemma}

\noindent{\bf Proof of Proposition \ref{extended}.}\
 By Lemma \ref{dp} (iii), we only need to prove the result for  $n\in \N.$  Take $E\in \Sigma_I$ and assume $t_K(E)=0$,
  then $A_{K+2}=B_{K+2}=I$ by the above lemma. Define
  \begin{equation*}
C:=\max\{\|T_l(E)\|, \|{\overline T}_{l}(E)\|: 1\le l\le 2^{K+2}\}.
\end{equation*}
Fix any $n\in \N$ and assume $n=m 2^{K+2}+l$, with $0\le l<2^{K+2}$. Then $T_n(E)=T_l(E)$ or ${\overline T}_{l}(E)$, since $A_{K+2}=B_{K+2}=I$. Thus \eqref{bd-transfer} follows.
\hfill $\Box$

\section{ $\Sigma_{II}$ and $\Sigma_{III}$ are dense and uncountable }\label{sec-exist-dense}

In this section, we will prove Theorem \ref{exist-dense}, i.e.,
both $\Sigma_{II}$ and $\Sigma_{III}$ are dense in $\sigma(H_\lambda)$ and uncountable.

Recall that in the Introduction we have defined a map
\begin{equation*}
f(x,y)=(x^2(y-2)+2,x^2y^2(y-2)+2)
\end{equation*}
and a strip
\begin{equation*}
S=\{(x,y)\in \mathbb{R}^2: |x|\le 1,\ y\le x^2-2 \}.
\end{equation*}

 As we have explained in the Introduction, it is crucial  to study the behavior of the  inverse branches of  $f$  on the  strip $S$.

\subsection{Dynamical properties of   $f$ }\

  Write  $(x_1,y_1)=f(x,y)$, the following relations will be frequently used later:
\begin{equation}\label{f-inv}
y_1-2=x^2y^2(y-2),\quad y_1-x_1^2+2=(y-x^2+2)(y-2)^2x^2.
\end{equation}
They follow directly from the definition of $f.$

There are four inverse branches for $f$, which are denoted as $F_{\epsilon,\eta}$ with $\epsilon,\eta\in\{+,-\}$,
\begin{equation*}
F_{\epsilon,\eta}(x,y)=\left( \epsilon\sqrt{\frac{2-x}{2-\eta\sqrt{\frac{2-y}{2-x}}}},\eta\sqrt{\frac{2-y}{2-x}}\right).
\end{equation*}
The corresponding domains are,
\[Dom(F_{\pm,-})=\{(x,y):x<2,y< 2\},\quad Dom(F_{\pm,+})=\mathcal{U},\]
where
$$\begin{array}{rcl}\mathcal{U}&=&\{(x,y):x<2, y<2, y-4x+6>0\}\bigcup\\
&&\{(x,y):x>2, y>2, y-4x+6>0\}.\end{array}$$
Write
$$
F_0=F_{-,-},\ \  F_1=F_{+,-},\ \  F_2=F_{-,+} \ \ \text{ and }\ \  F_3=F_{+,+}.
$$
We have the following basic lemma(see Figure \ref{pic-sigma-3}):
\begin{lemma}\label{S-inclusion}
 $F_2(\mathcal{U})\cap S=F_3(\mathcal{U})\cap S=\emptyset$.  $F_0(S), F_1(S)\subset S$ and $F_0(S)\cap F_1(S)=\emptyset.$
\end{lemma}
\begin{proof}
Since $S\subset [-1,1]\times (-\infty,-1]$, the first equation follows easily.

Next we show   $F_1(S)\subset S$.
Notice that  $S\subset Dom(F_1)$. Take any $(x,y)\in S$, denote $F_1(x,y)=(\hat x,\hat y)$, then
$$-\hat y=\sqrt{\frac{2-y}{2-x}}\ge\sqrt{\frac{4-x^2}{2-x}}=\sqrt{2+x}\ge1\ge|x|.$$
Thus we get
$$
0< \hat x=\sqrt{\frac{2-x}{2-\hat y}}\le 1.
$$
 On the other hand, by direct computation,
\begin{equation*}\label{minv}
\hat y-\hat x^2+2=\frac{y-x^2+2}{(2+\sqrt{\frac{2-y}{2-x}})(2-x)}\le 0,
\end{equation*}
so we get $(\hat x,\hat y)\in S$.  Thus $F_1(S)\subset S\cap \{(x,y):x>0\}$.

Since $S$ is symmetric w.r.t. $y$ axis and  $F_0(S)$ is the reflection of $F_1(S)$ w.r.t. $y$ axis, we get $F_0(S)\subset S\cap \{(x,y):x<0\}.$  So $F_0(S)\cap F_1(S)=\emptyset.$
\end{proof}

Recall that we have defined
$$
\Lambda(S)=\bigcap_{n\ge 0} f^{-n}(S).
$$
Define the parabola
 \begin{equation*}
\mathscr{P}:=\{(x,y): y=x^2-2\}.
\end{equation*}

\begin{lemma}\label{repeller}
Write $(x_k,y_k):=f^k(x,y)$ for any $k\ge0.$ If $(x,y)\in \Lambda(S)\setminus \mathscr{P},$
then
$$
\lim_{k\rightarrow\infty}x_k=0,\quad\lim_{k\rightarrow\infty}y_k=-\infty.
$$
\end{lemma}

\begin{proof}
For any $k\ge0$, write  $a_k=x_k^2-y_k-2$. Since
$$
(x,y)\in \Lambda(S)\setminus \mathscr{P}\subset S\setminus \mathscr{P},
$$ we have  $a_0>0.$ By the definition of $f$, we have $x_k=x_{k-1}^2(y_{k-1}-2)+2$. Then
by \eqref{f-inv},
$$
a_k=a_{k-1}(2-y_{k-1})^2x_{k-1}^2=a_{k-1}(2-y_{k-1})(2-x_k).
$$
Since $(x,y)\in \Lambda(S),$ we have $(x_k,y_k)=f^k(x,y)\in S.$
Thus $|x_k|\le1$ and  $y_k\le -1$. Hence $a_k\ge 3^ka_0 $.

Since $y_k=x_k^2-2-a_k$ and $|x_k|\le1$, we have
$\lim_{k\rightarrow\infty}y_k=-\infty.$

Since  $x_k^2=(2-x_{k+1})/(2-y_k)$, we have
$\lim_{k\rightarrow\infty}x_k=0.$
\end{proof}

By Lemma \ref{S-inclusion},
$$
\Lambda(S)= \bigcap_{n\ge 1} \bigcup_{w\in \{0,1\}^n} F_{w_1}\circ\cdots \circ F_{w_n}(S).
$$
For $\omega\in \{0,1\}^\infty$, define
$$
\Lambda_\omega:=\bigcap_{n\ge 1} F_{\omega_1}\circ\cdots \circ F_{\omega_n}(S).
$$
Then
$$
\Lambda(S)=\bigcup_{\omega\in \{0,1\}^\infty}\Lambda_\omega.
$$

\begin{lemma}\label{fiber}
For every $\omega\in \{0,1\}^\infty,$
 $p_\omega:=\Lambda_\omega\cap \mathscr{P}\ne\emptyset$ and  $p_\omega$ is connected.
 Moreover the connected component of $\Lambda_\omega$ which contains $p_\omega$ is  unbounded.  $\Lambda_\omega\cap \Lambda_{\tilde \omega}=\emptyset$ if $\omega\ne\tilde \omega.$
\end{lemma}

\begin{proof}
Write $ \mathscr{P}_1=\{(x,y): y=x^2-2, \ -1\le x\le1\},$ then $\mathscr{P}_1= \mathscr{P}\cap S$. By the definition of $F_0$ and $F_1$, it is direct to check that
$$
F_0(\mathscr{P}_1), \ F_1(\mathscr{P}_1)\subset \mathscr{P}_1\ \ \text{ and }\ \ F_0(S\setminus\mathscr{P}_1), \  F_1(S\setminus\mathscr{P}_1) \subset S\setminus\mathscr{P}_1.
$$
Consequently
$$
p_\omega=\left(\bigcap_{n\ge 1} F_{\omega_1}\circ\cdots \circ F_{\omega_n}(S)\right)\cap \mathscr{P}=\bigcap_{n\ge 1} F_{\omega_1}\circ\cdots \circ F_{\omega_n}(\mathscr{P}_1).
$$
Since $F_{\omega_1}\circ\cdots \circ F_{\omega_n}(\mathscr{P}_1)$ is compact, connected and decreasing, $p_\omega$ is nonempty, compact and connected.

Let $\Psi: \mathbb{S}^2\setminus\{N\} \to \R^2$ be   the standard polar projection from the unit sphere $\mathbb{S}^2$ to $\R^2$, where $N=(0,0,1)$ is the north pole.  Let $\Phi:\R^2\to \mathbb{S}^2\setminus\{N\}$ be the inverse of $\Psi.$ Define $\hat S=\Phi(S)\cup\{N\}\subset \mathbb{S}^2,$ then $\hat S$ is compact and connected. For $i=0,1$, define $\hat F_i: \hat S\to\hat S$ as
$$
\hat F_i(x)=
\begin{cases}
\Phi\circ F_i\circ\Psi(x)&x\ne N,\\
N&x=N.
\end{cases}
$$
 From the definition of $F_i$ and $\hat F_i$, it is easy to check that $\hat F_0,\hat F_1$ are continuous. Moreover, by Lemma \ref{S-inclusion},
$$
\hat F_i(\hat S\setminus\{N\})\subset \hat S\setminus\{N\}, \ \ (i=0,1) \ \ \text{ and }\ \ \hat F_0(\hat S)\cap \hat F_1(\hat S)=\{N\}.
$$
For any $\omega\in \{0,1\}^\infty,$  define
$$
\hat \Lambda_\omega:=\bigcap_{n\ge 0}\hat F_{\omega_1}\circ\cdots\circ \hat F_{\omega_n}(\hat S).
$$
Then $\hat \Lambda_\omega$ is  compact, connected and $N\in  \hat \Lambda_\omega.  $ Moreover from the definition of $\hat F_i$, we have
$
\hat \Lambda_\omega\setminus \{N\}= \Phi(\Lambda_\omega).
$
It is known that  if $E$ is a compact and connected set in $\mathbb{S}^2$ and $N\in  E$, then each connected component of $\Psi(E\backslash \{N\})$ is unbounded(see for example \cite{Newman} page 84). Applying to our situation, we conclude that each connected components of $\Lambda_\omega$ is unbounded. In particular, the component containing $p_\omega$ is unbounded.

 Now assume $\omega\ne \tilde \omega.$ Notice that $F_1, F_2$ are all injective, and by Lemma \ref{S-inclusion}, $F_0(S)\cap F_1(S)=\emptyset,$ then it is seen that  $\Lambda_\omega\cap \Lambda_{\tilde \omega}=\emptyset.$
\end{proof}

\subsection{  Type-II and type-III energies are dense and uncountable }\

The following initial condition is useful in what follows.

\begin{lemma}\label{initial-condition}
We have
$$
\begin{cases}
t_1^2(E)-t_2(E)-2=4\lambda^2&  E\in \R,\\
t_2^2(E)-t_3(E)-2>0&E\in {\rm int}(\sigma_1\cup\sigma_2).
\end{cases}
$$
\end{lemma}

\begin{proof}
The first equation follows trivially  from \eqref{t12}.

By \eqref{t12},\eqref{recurrence} and direct computation,  we know that $(t_2(E), t_3(E))$ is in the parabola $y=x^2+4\lambda^2x-8\lambda^2-2.$
Thus   $t_2^2(E)-t_3(E)-2>0$ if  $t_2(E)<2$.
By \eqref{t12}, we have
\begin{eqnarray*}
\sigma_1&=&\{E: |\lambda|\le |E|\le \sqrt{4+\lambda^2}\}\\
\{E: t_2(E)\le 2\}&=&\{E: \sqrt{1+\lambda^2}-1\le |E|\le  \sqrt{1+\lambda^2}+1\}.
\end{eqnarray*}
Since $\sigma_2=\{E: |t_2(E)|\le 2\}$, $\sqrt{1+\lambda^2}-1\le |\lambda|$ and  $\sqrt{4+\lambda^2}\le \sqrt{1+\lambda^2}+1$,  we conclude that $\sigma_1\cup\sigma_2\subset\{E: t_2(E)\le 2\} $ and consequently ${\rm int}(\sigma_1\cup\sigma_2)\subset  \{E:t_2(E)<2\}$.
\end{proof}

Now we relate type-II and type-III energies to $\Lambda(S)$. For any $k\ge 1$, write
$$
 \varphi_k(E)=(t_{k}(E),t_{k+1}(E)).
$$
By \eqref{recurrence}, it is direct to check that    for any $E\in\mathbb{R}$ and any $k\ge1$,
\begin{equation}\label{recurrent}
f(\varphi_k(E))=\varphi_{k+2}(E).
\end{equation}

\begin{lemma}\label{char-II-III}
We have
\begin{eqnarray*}
\Sigma_{II}&=&\bigcup_{k \text{ even}}\{E: \varphi_k(E)\in \Lambda(S)\setminus\mathscr{P}\}=\{E:t_{2n}(E)\to 0; t_{2n-1}(E)\to-\infty\} \\
\Sigma_{III}&=&\bigcup_{k \text{ odd}}\{E: \varphi_k(E)\in \Lambda(S)\setminus\mathscr{P}\}=\{E:t_{2n-1}(E)\to 0; t_{2n}(E)\to-\infty\}.
\end{eqnarray*}
\end{lemma}

\begin{proof} We only show the first equation, since the proof of the second one is the same.

At first we fix $E\in \Sigma_{II}$.
By the definition of type-II energy,  there exists an even number $k>0$ such that, for any $n\ge0$,
\begin{equation*}\label{condition-trace}
|t_{k+2n}(E)|\le 2\ \ \text{ and } \ \ |t_{k+2n+1}(E)|>2.
\end{equation*}
By Lemma \ref{basic} (i) and  (iii), we have $t_{k+2n+1}(E)<-2$ for any $n\ge 0$. So $0\le 2-t_{k+2n+2}(E)\le 4$ and $2-t_{k+2n+1}(E)>4$.
By \eqref{recurrence}, we have
$$
2-t_{k+2n+2}(E)=t_{k+2n}^2(E)(2-t_{k+2n+1}(E)).
$$
This implies that $|t_{k+2n}(E)|<1$ for any $n\ge 0$.  That is, for any $n\ge 0$,
$$
\varphi_{k+2n}(E)=(t_{k+2n}(E),t_{k+2n+1}(E))\in [-1,1]\times(-\infty,-2)\subset S.
$$
Since $\varphi_{k+2n}(E)=f^{n}(\varphi_{k}(E))$ by \eqref{recurrent}, we conclude that $\varphi_{k}(E)\in \Lambda(S)$. We also have  $\varphi_{k}(E)\not\in \mathscr{P}$ since $\varphi_{k}(E)\in [-1,1]\times (-\infty,-2)$ and $y\ge -2$ for any $(x,y)\in \mathscr{P}.$
Thus
$
 \varphi_{k}(E)\in \Lambda(S)\setminus\mathscr{P}.
$

Next we fix $k$ even, and take $E$ such that $\varphi_k(E)\in \Lambda(S)\setminus\mathscr{P}.$ By Lemma \ref{repeller} and the fact that $\varphi_{k+2n}(E)=f^{n}(\varphi_k(E))$, we have
$
t_{k+2n}(E)\to 0$ and $ t_{k+2n+1}(E)\to-\infty.
$

Finally if $E$ is such that $t_{2n}(E)\to 0$ and $ t_{2n-1}(E)\to-\infty$,
  by the definition of type-II energy, we conclude that $E\in\Sigma_{II}.$
\end{proof}

We are ready for the proof of Theorem \ref{exist-dense}.
\bigskip

\noindent{\bf Proof of Theorem \ref{exist-dense}}.\
Fix  $E_\ast\in\sigma(H_\lambda)$  and $\varepsilon>0$.
Since $\Sigma_I$ is dense in $\sigma(H_\lambda)$(see \cite{LQ}), there exists $E_0\in \Sigma_I$ such that $|E_0-E_\ast|<\varepsilon/2$.
By the definition of type-I energy, there exists $K\in \N$ such that $t_K(E_0)=0$. Then $t_n(E_0)=2$ for any $n\ge K+2$ by Lemma \ref{basic} (i). It is known that $\sigma_n$ is made of $2^n$ non-overlapping closed bands and $E_0$ is  an end point of some band $B_n\subset \sigma_n$, moreover $t_n$ is monotone on $B_n$ and $t_n(B_n)=[-2,2]$. By choosing odd (even) number $k> K+2$ large enough, we can assume $E_0$ is the endpoint of $B_{k}\subset \sigma_{k}$ and $B_{k}$ has length less than $ \epsilon/2$(see \cite{LQ} Appendix). Let $ E_1,E_2\in B_{k}$ be such that $t_{k}([E_1,E_2])=[-1,1]$.

\noindent {\bf Claim:} \  For any $\omega\in\{0,1\}^\infty$,
$$
\varphi_k([E_1,E_2])\cap \Lambda_\omega\ne \emptyset.
$$
\noindent $\lhd$ Let us show that the curve $\varphi_k([E_1,E_2])\subset S$ stays below $\mathscr{P}$ and touches the left and right boundaries of $S$(see Figure \ref{pic-sigma-3}) .

Fix $E\in [E_1,E_2]$.  Since $f(t_{n}(E),t_{n+1}(E))=(t_{n+2}(E),t_{n+3}(E))$,
by \eqref{f-inv},
\begin{eqnarray*}
&&t_{k}^2(E)-t_{k+1}(E)-2\\
&=&
\begin{cases}
(t_1^2(E)-t_2(E)-2) {\displaystyle\prod_{l=1}^{(k-1)/2}}t_{2l-1}^2(E)(2-t_{2l}(E))^2& k \text{ odd},\\
(t_2^2(E)-t_3(E)-2) {\displaystyle\prod_{l=1}^{k/2-1}}t_{2l}^2(E)(2-t_{2l+1}(E))^2& k \text{ even}.
\end{cases}
\end{eqnarray*}

Notice that $t_{k}(E)\ne 2$ since $E\in [E_1,E_2]$.  By Lemma \ref{basic} (ii),
  $t_i(E)\ne2$ for any $2\le i\le k$,
and  by Lemma \ref{basic} (i), $t_i(E)\ne0$ for any $1\le i\le k-2$.
Moreover, by Lemma \ref{initial-condition}, $t_1^2(E)-t_2(E)-2>0$ and  $t_2^2(E)-t_3(E)-2>0$ since $\sigma_k\subset \sigma_1\cup\sigma_2$ and thus
$$
E\in {\rm int}(B_k)\subset {\rm int}(\sigma_k)\subset {\rm int}(\sigma_1\cup\sigma_2).
$$
(see \cite{LQ}).    Hence $t_{k}^2(E)-t_{k+1}(E)-2>0$, i.e. $\varphi_k(E)$ is below  $\mathscr{P}$. On the other hand, since $t_{k} $ is monotone on $[E_1,E_2]$ and $\{t_{k}(E_1),t_{k}(E_2)\}=\{-1,1\}$,  the curve $\varphi_k([E_1,E_2])$ stays in $S$ and touches the left and right boundaries.

By Lemma \ref{fiber}, $\Lambda_\omega\subset S$ has a   unbounded connected component which is rooted in $\mathscr{P}$, then the curve $\varphi_k([E_1,E_2])$ must intersect $\Lambda_\omega$.
\hfill $\rhd$

\medskip

For any $\omega$, fix some $E_\omega\in [E_1,E_2]$ such that $\varphi_k(E_\omega)\in \Lambda_\omega,$ then $\varphi_k(E_\omega)\not\in \mathscr{P}$ since $\varphi_k([E_1,E_2])$ stays below $\mathscr{P}.$ Moreover, $E_\omega\ne E_{\tilde\omega}$ for $\omega\ne \tilde\omega$ since $\Lambda_\omega$ and $\Lambda_{\tilde \omega}$ are disjoint and $t_{k}$ is monotone on $[E_1,E_2]$.   Then  by Lemma \ref{char-II-III}, $E_\omega\in \Sigma_{III}(\Sigma_{II})$. We  also have
$$
|E_\omega-E_\ast|\le |E_\omega-E_0|+|E_0-E_\ast|\le \epsilon.
$$
Thus $\Sigma_{III}(\Sigma_{II})$ is dense in $\sigma(H_\lambda)$ and uncountable.
\hfill $\Box$

\bigskip

For later use, we show the following relation between $\nu_n(E)$ and $\omega_n(E)$ for type-II and type-III energies(see  \eqref{A-B-even} for the definitions of $\nu_n(E)$ and $\omega_n(E)$).

\begin{lemma}\label{energy-coupling inequality}
For any $E\in \Sigma_{II}\cup\Sigma_{III}$,  $\nu_n(E)\ne 0.$ Moreover,
\begin{equation}\label{kappa}
 \frac{\omega_{n}(E)}{\nu_{n}(E)}=\frac{E^2-\lambda^2}{2E}=:\kappa(E)
\end{equation}
with $|\kappa(E)|<1.$ We take $\theta(E)\in (-\pi/2,\pi/2)$  such that $\sin \theta(E)=\kappa(E)$.
\end{lemma}

\begin{proof}
At first we show  that $\nu_n(E)\ne 0$ for all $n\ge1.$
Since $E$ is in the spectrum, $E\ne 0$ by Lemma \ref{basic} (iv). By Lemma \ref{recur-mu-nu-omega}, $\nu_1(E)=4\lambda E\ne 0$ and
\[
\nu_{n+1}(E)=(t_{2n+1}(E)-2)t_{2n}(E)\nu_{n}(E).
\]
Since  E is of type-II or type-III, by Lemma \ref{basic} (i) and (ii),
$t_{n}(E)\neq 0$ for $n\ge 1$ and $t_n(E)\ne 2$ for $n\ge 2$. Consequently $\nu_n(E)\ne 0$ for all $n\ge 1$.

\eqref{kappa} follows directly from Lemma \ref{recur-mu-nu-omega}.

Next, we show  $|\kappa(E)|<1.$
By \eqref{A-B-even} and \eqref{kappa},  we have
\begin{equation*}\label{t-nu}
t_{2n+1}(E)= t_{2n}^2(E)+(\kappa^2(E)-1)\nu_{n}^2(E)-2.
\end{equation*}
By Lemma \ref{char-II-III},
$t_{2n}(E)\to0$ and $t_{2n+1}(E)\to-\infty$ if E is of type-II; $t_{2n+1}(E)\to0$ and $t_{2n}(E)\to-\infty$ if E is of type-III.
In both cases, it is easy to see that    $|\kappa(E)|<1.$
\end{proof}

\section{Type-II energy}\label{sec-type-II}

In this section, we prove Theorem \ref{main-norm} for type-II energy and Theorem \ref{main-type-II}.

 In Section \ref{4.1}, we study the asymptotic behaviors  of $t_n, \mu_n, \nu_n$ and $\omega_n$. Based on these asymptotic behaviors,   we obtain  the structures of $A_n$ and $B_n$ in Section \ref{4.2}.
  By these structures of $A_n$ and $B_n$ and an interpolation argument,  we   prove Theorem \ref{main-norm} for type-II energies in Section \ref{4.4}.
To   find the boundary condition of  the subordinate solution to the eigen-equation $H_\lambda\psi=E\psi$, we  study the stable directions of $A_{2n}$ and $B_{2n}$ in Section \ref{4.3}. Finally in Section \ref{4.5}, we   prove Theorem \ref{main-type-II}.

Throughout this section, we fix  $E\in \Sigma_{II}$.

\subsection{Exact asymptotic behaviors  of $t_n, \mu_n, \nu_n$ and $\omega_n$  }\label{4.1}

\begin{thm}\label{asy-t-mu-II}
There exists  $\gamma=\gamma(E)>0$ such that
\begin{equation}\label{t-mu-II}
\begin{cases}
{\displaystyle
\lim_{n\rightarrow\infty}\frac{|t_{2n-1}(E)|}{e^{2^{n}\gamma}}=\frac{1}{2},\
\lim_{n\rightarrow\infty}\frac{|t_{2n}(E)|}{e^{-2^{n}\gamma}}=2,\ \lim_{n\rightarrow\infty}\frac{|\mu_{n}(E)|}{e^{2^{n}\gamma}}=\frac{1}{2},}\\
{\displaystyle \lim_{n\rightarrow\infty}\frac{|\nu_{n}(E)|}{e^{2^{n}\gamma}}=\frac{\sec\theta(E)}{\sqrt{2}},\ \ \ \ \ \ \ \   \lim_{n\rightarrow\infty}\frac{|\omega_{n}(E)|}{e^{2^{n}\gamma}}=\frac{|\tan\theta(E)|}{\sqrt{2}},}
\end{cases}
\end{equation}
where $\theta(E)$ is defined in Lemma \ref{energy-coupling inequality}.
\end{thm}

 \begin{proof}
Since $E\in \Sigma_{II}$,
by Lemma \ref{char-II-III},
\begin{equation}\label{0-infty}
\lim_{n\rightarrow\infty}t_{2n}(E)=0\ \ \ \text{ and }\ \ \ \lim_{n\rightarrow\infty}t_{2n-1}(E)=-\infty.
\end{equation}
   Write
$$
|t_{2n}(E)|=e^{-2^n\gamma_{2n}}\ \ \ \text{ and }\ \ \
t_{2n-1}(E)=-e^{2^n\gamma_{2n-1}},
$$
then $\gamma_n>0$ for $n$ large enough.

\smallskip

\noindent {\bf Claim:} $\lim_{n\to\infty}\gamma_n$ exists.

\smallskip

\noindent $\lhd$ By \eqref{recurrence}, we have $t_{2n}(E)=t_{2n-2}^2(E)(t_{2n-1}(E)-2)+2.$
By taking the limits on both sides and using \eqref{0-infty}, we get
\begin{equation}\label{PreciseRelation1}
\lim_{n\rightarrow\infty}2^{n}(\gamma_{2n-1}-\gamma_{2n-2})=\ln 2.
\end{equation}

By \eqref{recurrence}, we also have $t_{2n+1}(E)=t_{2n-1}^2(E)(t_{2n}(E)-2)+2.$ Thus
$$
\frac{t_{2n+1}(E)}{t_{2n-1}^2(E)}=(t_{2n}(E)-2)+\frac{2}{t_{2n-1}^2(E)}.
$$
By taking the limits on both sides and using \eqref{0-infty}, we get
\begin{equation}\label{PreciseRelation2}
\lim_{n\rightarrow\infty}2^{n+1}(\gamma_{2n+1}-\gamma_{2n-1})=\ln 2.
\end{equation}

\eqref{PreciseRelation1} implies that $\gamma_{2n+1}-\gamma_{2n}=O(2^{-n})$.
\eqref{PreciseRelation2} implies that $\gamma_{2n+1}-\gamma_{2n-1}=O(2^{-n})$.
Consequently $\gamma_{2n}-\gamma_{2n-1}=O(2^{-n}).$ As a result,
$
\gamma_{n+1}-\gamma_n= O(2^{-n/2}).
$
Hence $\{\gamma_n: n\ge 1\}$ is a Cauchy sequence, and the limit exists. We denote the limit by $\gamma.$
\hfill $\rhd $

Now fix any small
$ \epsilon>0$. By \eqref{PreciseRelation2}, there exists $N \in \N$ such that for all $n>N$,
\[
\frac{(1-\epsilon)\ln 2}{2^{n+1}}\leq \gamma_{2n+1}-\gamma_{2n-1}\leq \frac{(1+\epsilon)\ln 2}{2^{n+1}}.
\]
Then we have
\[
\frac{(1-\epsilon)\ln 2}{2^{n}}\leq \gamma-\gamma_{2n-1}\leq \frac{(1+\epsilon)\ln 2}{2^{n}},
\]
or equivalently,
\[
\frac{1}{2^{1+\epsilon}}\le  \frac{|t_{2n-1}(E)|}{e^{2^{n}\gamma}}\le\frac{1}{2^{1-\epsilon}}.
\]
Since $\epsilon>0$ is arbitrary, we get the first equality of
\eqref{t-mu-II}.

By \eqref{0-infty}, $|t_{2n-1}(E)|\to \infty$, then  the first equality of \eqref{t-mu-II} implies that $\gamma>0.$

By \eqref{PreciseRelation1}, we have
$
e^{2^n\gamma_{2n+1}}|t_{2n}(E)|\to \sqrt{2}.
$
By the first equality of \eqref{t-mu-II}, we have
$
e^{2^{n}\gamma-2^{n}\gamma_{2n+1}}\to \sqrt{2}.
$
Combine these two equations, we get the second equality of \eqref{t-mu-II}.

 By \eqref{A-B-odd}, we have  $t_{2n}(E)=t_{2n-1}^2(E)-\mu_{n}^2(E)-2$. By \eqref{0-infty}, we conclude that
 $|\mu_n(E)|/|t_{2n-1}(E)|\to 1$.  Then the third equality of \eqref{t-mu-II} follows from the first one of \eqref{t-mu-II}.

 By \eqref{A-B-even} and \eqref{kappa}, we have  $t_{2n+1}(E)=t_{2n}^2(E)-(1-\kappa^2)\nu_{n}^2(E)-2$. By \eqref{0-infty}, we conclude that
 $\nu_n^2(E)/|t_{2n+1}(E)|\to 1/(1-\kappa^2)$.  Then the fourth and fifth equalities of \eqref{t-mu-II} follow from the first one of \eqref{t-mu-II} and \eqref{kappa}.
\end{proof}

To study the structures of $A_n$ and $B_n$, we need  more precise descriptions of $t_{n}(E) $ and $\mu_n(E):$

\begin{prop}\label{more-precise-II}
We have the following expansions:
\begin{equation}\label{t-prod-II}
\begin{array}{rcl}
e^{2^n\gamma} |t_{2n}(E)|&=&2+ O(e^{-2^n\gamma}),\\
e^{-2^n\gamma} |t_{2n-1}(E)|&=&\frac{1}{2}+O(e^{-2^n\gamma}),\\
e^{-2^n\gamma} |\mu_{n}(E)|&=&\frac{1}{2}+O(e^{-2^n\gamma})
\end{array}
\end{equation}
\end{prop}

\begin{proof}
 By \eqref{t-mu-II}, if we write
\begin{equation}\label{delta}
|t_{2n}|=2 e^{-2^{n}\gamma+\zeta_{2n}},\ \ \text{ and }\ \
|t_{2n-1}|=\frac{e^{2^{n}\gamma+\zeta_{2n-1}}}{2},
\end{equation}
then $\zeta_n\to 0$ when $n\to\infty.$

By \eqref{recurrence}, we have
 $
 t_{2n+1}=t_{2n-1}^2(t_{2n}-2)+2.
 $
By \eqref{delta},  it is equivalent to
$$e^{\zeta_{2n+1}-2\zeta_{2n-1}}=1\pm e^{-2^n\gamma+\zeta_{2n}}-4e^{-2^{n+1}\gamma-2\zeta_{2n-1}}
$$
Consequently,
$
|2\zeta_{2n-1}-\zeta_{2n+1}|\le 3 e^{-2^n\gamma }
$ 
for   $n$ large enough.
 Hence,
\begin{equation}\label{delta-1}
|\zeta_{2n-1}|\le 3 e^{-2^n\gamma}.
\end{equation}
This implies the second equation of \eqref{t-prod-II}.

Still by \eqref{recurrence}, we have
 $
 t_{2n+2}=t_{2n}^2(t_{2n+1}-2)+2.
 $
By \eqref{delta}, we have
\begin{eqnarray*}
e^{2\zeta_{2n}+\zeta_{2n+1}}=1\pm e^{-2^{n+1}\gamma+\zeta_{2n+2}}-4 e^{-2^{n+1}\gamma+2\zeta_{2n} }.
\end{eqnarray*}
Consequently,
$
|2\zeta_{2n}+\zeta_{2n+1}|\le 3 e^{-2^n\gamma }
$ for $n$ large enough. 
Then by \eqref{delta-1},
$
|\zeta_{2n}|\le 3 e^{-2^n\gamma }.
$ 
This implies the first equation of \eqref{t-prod-II}.

By \eqref{A-B-odd} and \eqref{t-mu-II}, for $n$ large enough,
$$
\big||t_{2n-1}|-|\mu_n|\big|=\frac{|2+t_{2n}|}{|t_{2n-1}|+|\mu_n|}\le 3e^{-2^n\gamma}.
$$
This implies  the third equation of \eqref{t-prod-II}.
\end{proof}

\subsection{Structures of $A_n$ and $B_n$}\label{4.2}\

Write  $\epsilon_n:={\rm sign}(\mu_n)$.
By Lemma \ref{recur-mu-nu-omega} and Lemma \ref{char-II-III}, $\epsilon_n$ is  eventually  constant, which we denote   by $\eta$.

\begin{thm}\label{exact-struc-odd}
  For $n$ large enough,
\begin{equation}\label{struc-AB-odd}
\frac{A_{2n+1}}{t_{2n+1}}=\frac{I- \eta U}{2}+O(e^{-2^{n}\gamma}), \quad
\frac{B_{2n+1}}{t_{2n+1}}=\frac{I+ \eta U}{2}+O(e^{-2^{n}\gamma}).
\end{equation}
\end{thm}
\begin{proof}
By \eqref{recur-AB-odd}, we have
\begin{equation}\label{div-AB-odd}
\frac{A_{2n+1}}{t_{2n+1}}=\frac{t_{2n}t_{2n-1}^2}{t_{2n+1}}\frac{A_{2n-1}}{t_{2n-1}}-
\frac{t_{2n-1}\mu_{n}}{t_{2n+1}}U-\frac{t_{2n-1}^2-1}{t_{2n+1}}I.
\end{equation}
By \eqref{t-prod-II}, there exist constants $c,d>0$ such that
$$\left\|{\frac{A_{2n+1}}{t_{2n+1}}}\right\|\le c+d\cdot  e^{-2^{n}\gamma}\left\|{\frac{A_{2n-1}}{t_{2n-1}}}\right\|.$$
By Lemma \ref{elementary},
$\left\|A_{2n+1}/t_{2n+1}\right\|=O(1).$
Now by  \eqref{div-AB-odd} and \eqref{t-prod-II},
$$\frac{A_{2n+1}}{t_{2n+1}}=\frac{I- \epsilon_n U}{2}+O(e^{-2^{n}\gamma}).$$
This proves the first equation of \eqref{struc-AB-odd}.

The proof of the second equation of \eqref{struc-AB-odd} is  same.
\end{proof}

 Since $E\in \Sigma_{II},$ $t_n(E)\ne 0$ for all $n\in \N$. Write  $\xi_n:=-{\rm sign}(t_{2n})$ and $\delta_n=\prod_{j=1}^n\xi_j$.

\begin{thm}\label{exact-struc-even}
There exist  nonzero constants $c=c(E)$ and $\hat c=\hat c(E)$ such that for $n$ large enough,
\begin{equation}\label{struc-AB-even}
\begin{array}{l}
t_{2n}A_{2n}=\delta_n c(V+\eta W)+O(e^{-2^{n+1}\gamma}),\\
t_{2n}B_{2n}=\delta_n \hat c(V-\eta W)+O(e^{-2^{n+1}\gamma}).
\end{array}
\end{equation}
\end{thm}

\begin{proof}
By \eqref{recur-AB-even}, we have
\begin{equation}\label{tA-even}
\begin{array}{rcl}
t_{2n}A_{2n}&=&t_{2n}t_{2n-1}t_{2n-2}A_{2n-2}+t_{2n}(1-t_{2n-2}^2)I\\
&&-t_{2n}t_{2n-2}\nu_{n-1}V-t_{2n}t_{2n-2}\omega_{n-1}W.\end{array}
\end{equation}

By \eqref{t-mu-II} and  \eqref{t-prod-II},
\begin{equation*}\label{norm-ite}
\|t_{2n}A_{2n}\| \le \left(1+O(e^{-2^n\gamma})\right)\|t_{2n-2}A_{2n-2}\|+O(e^{-2^n\gamma}).
\end{equation*}
By Lemma \ref{elementary},
$\|t_{2n}A_{2n}\|=O(1).$

On the other hand, \eqref{tA-even} implies that
$$
\delta_nt_{2n}A_{2n}= |t_{2n}t_{2n-1}|\delta_{n-1}t_{2n-2}A_{2n-2}+O(e^{-2^n\gamma}).
$$
By \eqref{t-prod-II}, $|t_{2n}t_{2n-1}|=1+O(e^{-2^n\gamma}).$ Thus
\begin{equation}\label{2-term}
\delta_nt_{2n}A_{2n}-\delta_{n-1}t_{2n-2}A_{2n-2}=O(e^{-2^n\gamma}).
\end{equation}
Consequently $\{\delta_nt_{2n}A_{2n}: n\ge 1\}$ is a Cauchy sequence.
Let  matrix $C$ be the limit, then \eqref{2-term} implies that
\begin{equation}\label{A}
t_{2n}A_{2n}= \delta_n C+O(e^{-2^{n+1}\gamma}).
\end{equation}

The same proof shows that there exists a matrix $\widehat C$ such that
\begin{equation}\label{B}
t_{2n}B_{2n}= \delta_n \widehat C+O(e^{-2^{n+1}\gamma}).
\end{equation}

In the following we show $C$ is nonzero.
If otherwise, we have $t_{2n}A_{2n}=O(e^{-2^{n+1}\gamma}),$ then  $\|A_{2n}\|=O(e^{-2^{n}\gamma})$ by \eqref{t-mu-II}.
On the other hand, \eqref{B} implies that $\|B_{2n}\|=O(e^{2^{n}\gamma})$.  Thus
$
\|A_{2n+1}\|\le \|B_{2n}\|\|A_{2n}\|=O(1).
$
However, by Theorem \ref{exact-struc-odd},  we have $\|A_{2n+1}\|\sim e^{2^{n+1}\gamma}$, which is a contradiction.
So $C$ is nonzero.

The same proof  shows that $\widehat C$ is nonzero.

To obtain the exact forms of $C$ and $\widehat C$, further informations are needed.  We will postpone the proof to Section \ref{4.3}.
\end{proof}

\begin{rem}\label{osci-II}
{\rm Recall that $A_{n}=T_{2^{n}}$.
Thus Theorem \ref{exact-struc-odd} implies that $\|T_k\|\sim e^{\sqrt{2}\gamma \sqrt{k}} $ for $k=2^{2n+1}$,
while Theorem \ref{exact-struc-even} implies that $\|T_k\|\sim e^{\gamma\sqrt{k}} $ for $k=2^{2n}$. This implies that there  exists fluctuation for the coefficients of $\sqrt{k}$.

}
\end{rem}

 Theorem \ref{exact-struc-even} has the following  consequence, which will be used in the proof of Theorem \ref{main-norm}.

\begin{cor}\label{exact-norm-II}
 We have
$$
\begin{cases}
{\displaystyle\lim_{n\to\infty}\frac{\|B_{2n}^2A_{2n}\|}{e^{2^n\gamma}}=\lim_{n\to\infty}\frac{\|A_{2n}B_{2n}^2\|}{e^{2^n\gamma}}
=\sqrt{1+c^2}},\\
{\displaystyle\lim_{n\to\infty}\frac{\|A_{2n}^2B_{2n}\|}{e^{2^n\gamma}}=\lim_{n\to\infty}\frac{\|B_{2n}A_{2n}^2\|}{e^{2^n\gamma}}
=\sqrt{1+\hat c^2}}.
\end{cases}
$$
\end{cor}

\begin{proof}
By Lemma \ref{recurrence-AB},  Theorem \ref{exact-struc-odd},  Theorem \ref{exact-struc-even} and \eqref{t-mu-II}, we have
\begin{eqnarray*}
&&B_{2n}^2A_{2n}\\
&=&t_{2n}A_{2n+1}-A_{2n}
=t_{2n}t_{2n+1}\frac{A_{2n+1}}{t_{2n+1}}-t_{2n}^{-1} t_{2n}A_{2n}\\
&=&t_{2n}t_{2n+1}\left(\frac{I-\epsilon_nU}{2}+O(e^{-2^n\gamma})\right)-t_{2n}^{-1}\left(\delta_n c(V+\epsilon_n W)+O(e^{-2^n\gamma})\right)\\
&=&\frac{e^{2^n\gamma}}{2}\left(\pm(I\pm U)\pm c(V\mp W)+o(1)\right).
\end{eqnarray*}
Then the first equality holds. The same proof shows that the other three equalities hold.
\end{proof}

\subsection{Proof of Theorem \ref{main-norm} for type-II energy}\label{4.4}\

We will  frequently use the following two facts: if $X\in {\rm SL}(2,\R)$, then $\|X\|=\|X^{-1}\|; $ if $X$ is invertible and $X=YZ$, then
$$
\|X\|\le\|Y\|\|Z\|\quad \text{ or equivalently }\quad \|Y\|\ge\|X\|/\|Z\|.
$$

Theorem \ref{main-norm} follows from the following lemma, the proof of which is by interpolation.  See also Figure \ref{pic-log-norm} for some intuition.

\begin{lemma}\label{norm-II}
There exists a constant $C=C(E)>1$ such that
\begin{equation}\label{T-k-II}
C^{-n-1}e^{2^n\gamma}\le\|T_{ k}\|, \|\overline{T}_k\|\le C^{n+1} e^{2^{n+1}\gamma}
\end{equation}
if $2^{2n}\le k<2^{2(n+1)}$.
Consequently, if $k<2^{2n},$ then
\begin{equation}\label{up-T-k-II}
 \|T_{ k}\| , \|\overline{T}_k\|\le C^{n} e^{2^{n}\gamma}.
\end{equation}
\end{lemma}

\begin{proof}
The proofs for $T_k$ and $\overline{T}_k$ are essentially the same, so we only prove the result for $T_k.$ \eqref{up-T-k-II} is a direct consequence of \eqref{T-k-II}, so we only need to show \eqref{T-k-II}.

By Theroem \ref{exact-struc-odd}, Theroem \ref{exact-struc-even} and   Corollary \ref{exact-norm-II}, there exists a constant $C=C(E)>1$ such that for any $n\in\N$,
\begin{equation}\label{D-II}
C^{-1} e^{2^n\gamma}\le \|D_{2n-1}\|,\ \|D_{2n}\|,\ \|D_{2n}\bar{D}_{2n}^2\|,\ \|D_{2n}^2\bar{D}_{2n}\|\le C e^{2^n\gamma},
\end{equation}
where $D\in \{A,B\}$ and  $\bar D=B$ if $D=A$ and  $\bar D=A$ if $D=B$.

Assume $2^{2n}\le k< 2^{2(n+1)}$. Write
$
k=\sum_{j=0}^n l_j 2^{2j}
$
with $0\le l_j\le 3$, then $l_n>0$.  Define
$\mathcal{P}(k):=\min\{j:l_j>0\}.$

If $\mathcal{P}(k)=n$, then $k=l_n 2^{2n}$ with $l_n=1$, $2$ or $3$.
By \eqref{D-II},
$$
\begin{cases}
C^{-1}e^{2^{n}\gamma}\le \|T_k\|\le Ce^{2^{n}\gamma},& l_n=1 \text{ or } 3;\\
C^{-1}e^{2^{n+1}\gamma}\le \|T_k\|\le Ce^{2^{n+1}\gamma},&l_n=2.
\end{cases}
$$
Then \eqref{T-k-II} holds.

Now assume $\mathcal{P}(k)<n$. We will prove that if $l_{\mathcal{P}(k)}=1$ or $3$, then
\begin{equation}\label{ind-norm-II}
C^{-n-1+\mathcal{P}(k)}e^{(2^n+2^{\mathcal{P}(k)})\gamma}\le
\|T_k\|\le C^{n+1-\mathcal{P}(k)}e^{(2^{n+1}-2^{\mathcal{P}(k)})\gamma};
\end{equation}
if $l_{\mathcal{P}(k)}=2$ then
\begin{equation}\label{ind-norm-2-II}
C^{-n-1+\mathcal{P}(k)}e^{2^n\gamma}\le \|T_k\|\le C^{n+1-\mathcal{P}(k)}e^{2^{n+1}\gamma},
\end{equation}
which obviously implies \eqref{T-k-II}.
We prove it by induction on $\mathcal{P}(k)$.

Consider firstly the case $\mathcal{P}(k)=n-1$.

If $l_n=1$ or $3$.
By $T_k=T_{l_n 2^{2n}\rightarrow k}T_{l_n 2^{2n}}$ and \eqref{D-II},
$$
\begin{array}{ll}
\|T_k\|\le C^2e^{(2^n+2^{n-1})\gamma}=C^2e^{(2^{n+1}-2^{n-1})\gamma},& \mbox{if } l_{n-1}=1 \mbox{ or } 3;\\
\|T_k\|\le C^2e^{(2^n+2^{n})\gamma}=C^2e^{2^{n+1}\gamma},& \mbox{if } l_{n-1}=2.
\end{array}
$$
By $T_{(l_n+1) 2^{2n}}=T_{k\rightarrow (l_n+1)2^{2n}}T_{k}$ and \eqref{D-II},
$$
\begin{array}{ll}
\|T_k\|\ge C^{-2}e^{(2^{n+1}-2^{n-1})\gamma}=C^{-2}e^{(2^{n}+2^{n-1})\gamma},& \mbox{if } l_{n-1}=1 \mbox{ or } 3;\\
\|T_k\|\ge C^{-2}e^{(2^{n+1}-2^{n})\gamma}=C^{-2}e^{2^{n}\gamma},& \mbox{if } l_{n-1}=2.
\end{array}
$$

If $l_n=2$. By $T_k=T_{k\rightarrow 3\times 2^{2n}}^{-1}T_{3\times 2^{2n}}$ and \eqref{D-II},
$$
\begin{array}{ll}
\|T_k\|\le C^2e^{(2^n+2^{n-1})\gamma}=C^2e^{(2^{n+1}-2^{n-1})\gamma},& \mbox{if } l_{n-1}=1 \mbox{ or } 3;\\
\|T_k\|\le C^2e^{(2^n+2^{n})\gamma}=C^2e^{2^{n+1}\gamma},& \mbox{if } l_{n-1}=2.
\end{array}
$$
By $T_{2\times 2^{2n}}=T_{2\times 2^{2n}\rightarrow k}^{-1}T_k$ and \eqref{D-II},
$$
\begin{array}{ll}
\|T_k\|\ge C^{-2}e^{(2^{n+1}-2^{n-1})\gamma}=C^{-2}e^{(2^{n}+2^{n-1})\gamma},& \mbox{if } l_{n-1}=1 \mbox{ or } 3;\\
\|T_k\|\ge C^{-2}e^{(2^{n+1}-2^{n})\gamma}=C^{-2}e^{2^{n}\gamma},& \mbox{if } l_{n-1}=2.
\end{array}
$$
Thus \eqref{ind-norm-II} and \eqref{ind-norm-2-II} holds for $\mathcal{P}(k)=n-1$.

Now take any $p<n-1$.
Suppose \eqref{ind-norm-II} and \eqref{ind-norm-2-II} hold for any $\tilde k$ with $\mathcal{P}(\tilde k)>p$.
Take any $k$ with $\mathcal{P}(k)=p$.

If $l_{p+1}=0$ or $2$, set $k'=k-l_p 2^{2p}+2^{2(p+1)}$, then $\mathcal{P}(k')=p+1$ and $l_{p+1}(k')=1$ or $3$.
By $T_{k'}=T_{k\rightarrow k'}T_k$, inductive hypothesis and \eqref{D-II},
$$
\begin{array}{ll}
\|T_k\|\le C^{n+1-p}e^{(2^{n+1}-2^{p+1}+2^p)\gamma}=C^{n+1-p}e^{(2^{n+1}-2^{p})\gamma},& \mbox{if } l_{p}=1 \mbox{ or } 3;\\
\|T_k\|\le C^{n+1-p}e^{(2^{n+1}-2^{p+1}+2^{p+1})\gamma}=C^{n+1-p}e^{2^{n+1}\gamma},& \mbox{if } l_{p}=2;\\
\|T_k\|\ge C^{-n-1+p}e^{(2^{n}+2^{p+1}-2^p)\gamma}=C^{-n-1+p}e^{(2^{n}+2^{p})\gamma},& \mbox{if } l_{p}=1 \mbox{ or } 3;\\
\|T_k\|\ge C^{-n-1+p}e^{(2^{n}+2^{p+1}-2^{p+1})\gamma}=C^{-n-1+p}e^{2^{n}\gamma},& \mbox{if } l_{p}=2.
\end{array}
$$

If $l_{p+1}=1$ or $3$, set $k'=k-l_p 2^{2p}$, then $\mathcal{P}(k')=p+1$ and $l_{p+1}(k')=1$ or $3$.
By $T_{k}=T_{k'\rightarrow k}T_{k'}$, inductive hypothesis and \eqref{D-II},
$$
\begin{array}{ll}
\|T_k\|\le C^{n+1-p}e^{(2^{n+1}-2^{p+1}+2^p)\gamma}=C^{n+1-p}e^{(2^{n+1}-2^{p})\gamma},& \mbox{if } l_{p}=1 \mbox{ or } 3;\\
\|T_k\|\le C^{n+1-p}e^{(2^{n+1}-2^{p+1}+2^{p+1})\gamma}=C^{n+1-p}e^{2^{n+1}\gamma},& \mbox{if } l_{p}=2;\\
\|T_k\|\ge C^{-n-1+p}e^{(2^{n}+2^{p+1}-2^p)\gamma}=C^{-n-1+p}e^{(2^{n}+2^{p})\gamma},& \mbox{if } l_{p}=1 \mbox{ or } 3;\\
\|T_k\|\ge C^{-n-1+p}e^{(2^{n}+2^{p+1}-2^{p+1})\gamma}=C^{-n-1+p}e^{2^{n}\gamma},& \mbox{if } l_{p}=2.
\end{array}
$$
Thus \eqref{ind-norm-II} and \eqref{ind-norm-2-II} holds for $\mathcal{P}(k)=p$. This finishes the proof.
\end{proof}

\noindent {\bf Proof of Theorem \ref{main-norm}(type-II energy).} \
By Lemma \ref{dp} (iii), we only need to consider $T_{n}.$
Fix $n\in \N$, let $k$ be the unique integer such that
$2^{2k}\le n<2^{2(k+1)}$. Then
$
2^k\le \sqrt{n}< 2^{k+1}.
$
By \eqref{T-k-II}, we have
\begin{equation}\label{T-n-lu}
 n^{-\alpha} e^{\frac{\gamma}{2}\sqrt{n}}
 \le C^{-k-1}e^{\frac{\gamma}{2}\sqrt{n}}\le \|T_n\|
 \le C^{k+1}e^{2\gamma\sqrt{n}}\le n^{\alpha} e^{2\gamma\sqrt{n}},
\end{equation}
where $\alpha=\log C/\log 2.$ Then Theorem \ref{main-norm} holds for any $c_2>2$ and $c_1\in (0,1/2)$.
\hfill $\Box$

\subsection{Stable directions of $\{A_{2n}\}$ and $\{B_{2n}\}$}\label{4.3}\

\begin{prop}\label{Oseledets-II}
There exist  unit vectors $s, \hat s\in \R^2$ such that
\begin{equation}\label{sub-decay-2n}
\|A_{2n}s\|,\ \|B_{2n}\hat s\|= O(e^{-2^{n}\gamma})
\end{equation}
and for any unit vector $v (\hat v)$ which is independent with $s (\hat s)$,
$$
\|A_{2n}v\|,\ \|B_{2n}\hat v\| \gtrsim e^{2^n\gamma }.
$$
 \end{prop}

 We only prove the result for $\{A_{2n}\}$, since the prove for $\{B_{2n}\}$ is the same. The proof is inspired by Oseledets' argument.

By the singular value decomposition, there exist four  unit vectors $s_n, u_n, v_n $ and $w_n$   such that
$$
s_n\perp u_n,\quad v_n\perp w_n,\quad A_{2n}s_n=\|A_{2n}\|^{-1}v_n,\quad
A_{2n}u_n=\|A_{2n}\|w_n.
$$

\begin{lemma}
The sequence $(s_n)_n$ is a Cauchy sequence in projective space. Let $s$ be the limit, then
\begin{equation}\label{angle-s-sn}
|\sin\measuredangle(s_n,s)|=O(e^{-2^{n+1}\gamma}).
\end{equation}
\end{lemma}

\begin{proof}
Write  $\beta_n=\measuredangle(s_n,s_{n+1})$.
Then $s_n=s_{n+1}\cos\beta_n + u_{n+1}\sin\beta_n.$
Thus
$$
A_{2n+2}s_n=\cos\beta_n \|A_{2n+2}\|^{-1}v_{n+1}+\sin\beta_n \|A_{2n+2}\|w_{n+1}.
$$
Recall  that $A_{2n+2}=A_{2n}B_{2n}^2 A_{2n}$, then
\begin{eqnarray*}
|\sin\beta_n|\le\frac{\|A_{2n+2}s_n\|}
{\|A_{2n+2}\|}=\frac{\|A_{2n}B_{2n}^2v_n\|}
{\|A_{2n}\|\|A_{2n+2}\|}.
\end{eqnarray*}
By Theorem \ref{exact-struc-even}, Corollary \ref{exact-norm-II} and
\eqref{t-mu-II},  we get
$
\|A_{2n}\|, \|A_{2n}B_{2n}^2\|\sim e^{2^n\gamma}.
$
Consequently,
\begin{eqnarray}\label{beta-n}
|\sin\beta_n|=O(\|A_{2n+2}\|^{-1})=O(e^{-2^{n+1}\gamma}).
\end{eqnarray}
This implies that $\{s_n\}_n$ is a Cauchy sequence.
Let $s$ be the limit, then \eqref{angle-s-sn} follows from \eqref{beta-n}.
\end{proof}

 Indeed $s$ is the direction we are looking for.

\noindent {\bf Proof of Proposition \ref{Oseledets-II}.}\
 Let $\theta_n$ be the angle between $s_n$ and $s$,
then $
|\sin\theta_{n}|=O(e^{-2^{n+1}\gamma})
$ by \eqref{angle-s-sn}.
Since  $s=\cos \theta_n  s_n+\sin\theta_n u_n,$ we have
\begin{eqnarray*}
\|A_{2n}s\|&=&\|\cos \theta_n A_{2n} s_n+\sin\theta_nA_{2n}u_n\|
\le \|A_{2n}\|^{-1}+|\sin \theta_n|\|A_{2n}\|\\
&=&O(e^{-2^{n}\gamma}).
\end{eqnarray*}

Now fix an  unit vector $v$ which is independent with $s$, denote the angle between $v$ and $ s_n$ in projective space by $\phi_n$.
Then  $\measuredangle(s,v)/2 \le \phi_n\le \pi/2$ for $n$ large enough.
Since $v=\cos\phi_n s_n+\sin \phi_n u_n,$ we have
$$
\|A_{2n}v\|=\|\cos\phi_n A_{2n}s_n+\sin \phi_n A_{2n}u_n\|\ge |\sin\phi_n|\|A_{2n}\|\gtrsim e^{2^n\gamma}.
$$
Thus the result follows.
\hfill $\Box$

We call $s$ and $\hat s$ the {\it stable directions} of $\{A_{2n}\}$ and $\{B_{2n}\}$, respectively. Recall that $v_\theta=(\cos\theta,-\sin\theta)^t.$

\begin{prop}\label{parallel-s-II}
  $s\parallel v_{-\frac{\eta \pi}{4}}$ and $\hat s\parallel v_{\frac{\eta \pi}{4}}$. Consequently
  \begin{equation}\label{s-hat-s}
  s\perp \hat s, \ \  Us=\eta s\ \ \text{ and } \ \  U\hat s=-\eta \hat s.
  \end{equation}
\end{prop}

\begin{proof}
At first we show that $s\parallel v_{-\frac{\eta \pi}{4}}.$
If otherwise, $s=v_\theta$ with  $\theta\ne -\frac{\eta \pi}{4}+k\pi$, then $\cos\theta+\eta\sin\theta\ne 0.$
On the other hand, by \eqref{struc-AB-odd},
$$
\frac{A_{2n-1}}{t_{2n-1}}=\frac{I- \eta U}{2}+O(e^{-2^{n-1}\gamma}).
$$
Consequently,
$$
\|A_{2n-1}s\|=\frac{|t_{2n-1}|}{2}\|(\cos\theta+\eta\sin\theta)(1,-\eta)^t+O(e^{-2^{n-1}\gamma}) \|\sim e^{2^{n}\gamma}.
$$
By \eqref{struc-AB-odd}, $\|B_{2n-1}\|\sim e^{2^{n}\gamma}.$ Since $\det B_{2n-1}=1,$ we have $\|B_{2n-1}^{-1}\|=\|B_{2n-1}\|$. Thus
\[
\|A_{2n}s\|=\|B_{2n-1}A_{2n-1}s\|\geq \frac{\|A_{2n-1}s\|}{\|B_{2n-1}^{-1}\|}=\frac{\|A_{2n-1}s\|}{\|B_{2n-1}\|}\sim 1,
\]
 which  contradicts with \eqref{sub-decay-2n}.

 The same proof shows that $\hat s\parallel v_{\frac{\eta \pi}{4}}$.

 Since $s\parallel (1,\eta)^t$ and $\hat s\parallel (1,-\eta)^t$, \eqref{s-hat-s} follows easily.
\end{proof}

Now we can finish the proof of Theorem \ref{exact-struc-even}.

\noindent {\bf Proof of Theorem \ref{exact-struc-even}(continued).}\
By \eqref{A}, for  $n$ large enough,
$t_{2n}A_{2n}= \delta_n C+O(e^{-2^{n+1}\gamma}).$
Write
$$
C=
\begin{pmatrix}
c_{11}&c_{12}\\
c_{21}&c_{22}
\end{pmatrix}.
$$
By Proposition  \ref{parallel-s-II},  $s\parallel (1,\eta)^t$. By Proposition \ref{Oseledets-II},
for   $n$ large enough,
$$
\begin{cases}
|t_{2n}^{-1}(c_{11}+\eta c_{12}+O(e^{-2^{n+1}\gamma}))|&=O(e^{-2^n\gamma})\\
|t_{2n}^{-1}(c_{21}+\eta c_{22}+O(e^{-2^{n+1}\gamma}))|&=O(e^{-2^n\gamma})
\end{cases}
$$
On the other hand, by  Remark \ref{rembasic}, we have
$
t_{2n}^{-1}(c_{21}+ c_{12} )=O(e^{-2^n\gamma}).
$
Since $t_{2n}\sim e^{-2^n\gamma}$, we conclude that
$$c_{11}=-\eta c_{12}, \ \ c_{21}=-\eta c_{22}\ \ \text{ and } \ \ c_{21}=-c_{12}.
$$
This means that $C=c(V+\eta W)$ for some $c$ nonzreo.

The same argument shows that $\widehat C=\hat c (V-\eta W)$ for some $\hat c$ nonzero.
\hfill $\Box$

\medskip

Now we study the behaviors of $A_{2n+1} s$ and $B_{2n+1}\hat s$.

\begin{prop}
Let  $c$ and $\hat c$ be given as in Theorem \ref{exact-struc-even}, then
\begin{equation}\label{const-limit-II}
A_{2n+1}s=\delta_n\hat c\hat s +O(e^{-2^{n+1}\gamma})\ \  \text{ and }\ \  B_{2n+1}\hat s=\delta_n c s +O(e^{-2^{n+1}\gamma}).
\end{equation}
\end{prop}

\begin{proof}
By  Proposition \ref{parallel-s-II}, without loss of generality, we take $s=(1,\eta)^t/\sqrt{2}$ and $\hat s=(1,-\eta)^t/\sqrt{2}$.  By Remark \ref{rembasic} and  \eqref{struc-AB-even}, there exist real $c_*, p, q $ and $ \hat c_*, \hat p, \hat q$ such that
\begin{eqnarray}
\label{trace}t_{2n}A_{2n}&=& c_\ast
\begin{pmatrix}
1&-\eta\\
\eta&-1
\end{pmatrix}+
\begin{pmatrix}
p&0\\
0&q
\end{pmatrix}
=2c_\ast s \hat s^t+
\begin{pmatrix}
p&0\\
0&q
\end{pmatrix},\\
\label{trace-B}
t_{2n}B_{2n}&=& \hat c_\ast
\begin{pmatrix}
1&\eta\\
-\eta&-1
\end{pmatrix}+
\begin{pmatrix}
\hat p&0\\
0&\hat q
\end{pmatrix}
=2\hat c_\ast \hat s s^t+
\begin{pmatrix}
\hat p&0\\
0&\hat q
\end{pmatrix}
\end{eqnarray}
with
\begin{equation}\label{est-eta-ep}
|\delta_n c-c_\ast|,\ |\delta_n\hat c-\hat c_\ast| =O(e^{-2^{n+1}\gamma}),\
 p, q, \hat p, \hat q=O(e^{-2^{n+1}\gamma}).
\end{equation}

By taking the trace on both sides of \eqref{trace}, we get $p+q=t_{2n}^2$.
Since $A_{2n+1}=B_{2n}A_{2n}$ and $s\perp \hat s$, by \eqref{trace}, \eqref{trace-B}, \eqref{est-eta-ep} and direct computation,
\begin{eqnarray*}
t_{2n}^2A_{2n+1}s&=&\hat c_\ast (p+q)\hat s+
\frac{1}{\sqrt{2}}
\begin{pmatrix}
p\hat p\\
\eta q\hat q
\end{pmatrix}
=\hat c_\ast t_{2n}^2 \hat s+ O(e^{-2^{n+2}\gamma}).
\end{eqnarray*}
Consequently by \eqref{t-mu-II} and \eqref{est-eta-ep}, we have
$$
A_{2n+1}s=\hat c_\ast  \hat s+ t_{2n}^{-2}O(e^{-2^{n+2}\gamma})
=\hat c_\ast  \hat s+ O(e^{-2^{n+1}\gamma})=\delta_n\hat c \hat s +O(e^{-2^{n+1}\gamma}).
$$

Similarly, we get
$
B_{2n+1}\hat s=\delta_n c s +O(e^{-2^{n+1}\gamma}).
$
\end{proof}

\subsection{Subordinate solution}\label{4.5}\

In this subsection, we prove Theorem \ref{main-type-II}. As a preparation, we have

\begin{lemma}\label{stab-globe-II}
There exist $\alpha=\alpha(E)>0$ and  unit vector $\hat v$ such that for any $n\in\N$ and  any  vector $w \nparallel \hat v$, we have
$$
\|T_n w\|\gtrsim n^{-\alpha} e^{\frac{\gamma}{2}\sqrt{n}}.
$$
\end{lemma}

\begin{proof}
By singular value decomposition,  there exist unit vectors $s_n$,
$u_n$, $x_n$ and $y_n$ such that $s_n\perp u_n$, $x_n\perp y_n$ and
$
T_ns_n=\|T_n\|^{-1}x_n, T_nu_n=\|T_n\|y_n.
$
Write $\beta_n=\angle(s_n,s_{n+1}),$ then
$
s_n=s_{n+1}\cos \beta_n+u_{n+1}\sin\beta_n.
$
Since  $T_{n+1}=T_{n\to n+1}T_n$, we have
\begin{equation*}
\|T_{n+1}s_n\|\ge \|T_{n+1}u_{n+1}\sin\beta_n\|=|\sin \beta_n| \|T_{n+1}\|\gtrsim|\sin \beta_n| \|T_{n}\| .
\end{equation*}
On the other hand
\begin{equation*}
\|T_{n+1}s_n\|=\|T_{n\to n+1}T_ns_n\|\lesssim\|T_ns_n\|\lesssim \|T_n\|^{-1}.
\end{equation*}
Combine with \eqref{T-n-lu}, we have
\begin{equation*}
|\sin \beta_n|\lesssim \|T_n\|^{-2}\lesssim n^{2\alpha} e^{-\gamma\sqrt{n}}.
\end{equation*}
This implies that $s_n$ converges to an unit vector $\hat v$ in the projective space.

Now assume unit vector $w$ is independent with $\hat v$, let $\varphi$ be the angle of $w$ and $\hat v$ in projective space, then  $0<\varphi<\pi/2$.
Let $\theta_n$ be the angle of $w$ and $s_n$ in projective space, then when $n$ large enough, we have
$
{\varphi}/{2}<\theta_n<{\pi}/{2}.
$
By $w=s_n\cos\theta_n+u_n\sin\theta_n$ and \eqref{T-n-lu}, we get
\begin{equation*}
\|T_nw\|\ge |\sin\theta_n|\|T_nu_n\|\ge
|\sin\frac{\varphi}{2}|\|T_n\|\gtrsim n^{-\alpha}e^{\frac{\gamma}{2}\sqrt{n}}.
\end{equation*}
Now for any $w\nparallel \hat v$, we have $\|T_n w\|=\|w\|\|T_n(w/\|w\|)\|\gtrsim n^{-\alpha}e^{\frac{\gamma}{2}\sqrt{n}}.$
\end{proof}

Recall that any solution of $H_\lambda\phi=E\phi$ is uniquely determined by $\vec{\phi}_0=(\phi_1,\phi_0)^t$.
We denote by $\psi^\theta$ the solution with $\vec{\psi}^\theta_0=v_\theta.$

\smallskip

\noindent {\bf Proof of Theorem \ref{main-type-II}.\ }
We will show that $\psi^{-\eta\pi/4}$ is the desired subordinate solution. By  Proposition \ref{parallel-s-II}, without loss of generality, we take $s=(1,\eta)^t/\sqrt{2}=v_{-\eta\pi/4}, $ where $s$ is the stable direction of $\{A_{2n}:n\ge1\}$.

At first we note that, if $\phi$ is a solution of $H_\lambda\phi=E\phi$, then by Lemma \ref{dp} (iii),  for any $n\in \N,$
\begin{equation}\label{U-phi-n}
U\vec{\phi}_{-n}=T_n U\vec{\phi}_0.
\end{equation}

Next we claim that  $\hat v\parallel s$, where $\hat v$ is from Lemma \ref{stab-globe-II}.
Indeed if on the contrary, by Lemma \ref{stab-globe-II},
$$
\|T_ns\|\gtrsim n^{-\alpha}e^{\frac{\gamma}{2}\sqrt{n}}\to \infty.
$$
On the other hand, by Proposition \ref{Oseledets-II} and Proposition \ref{parallel-s-II}, we have
$$
\|T_{2^{2n}}s\|=\|A_{2n}s\|=O(e^{-2^n\gamma})\to 0.
$$
 Thus we get a contradiction.

As a result, we have $\hat v\parallel s$.  Let $\psi=\psi^{-\eta\pi/4}$, then $\vec{\psi}_0=s$ and $\vec{\psi}_n=T_ns$. By \eqref{s-hat-s},  $U\vec{\psi}_0=\eta\vec{\psi}_0$. By \eqref{U-phi-n},  we have
\begin{equation}\label{symmetry}
\vec{\psi}_{-n}=\eta U\vec{\psi}_n.
\end{equation}
Thus  the second equation of \eqref{subord-II} follows from \eqref{sub-decay-2n} and \eqref{symmetry};
the third equation of \eqref{subord-II} follows from \eqref{s-hat-s}, \eqref{const-limit-II} and \eqref{symmetry}.
Assume $\phi$ is
independent with $\psi$. Since  $\vec{\psi}_0=s=v_{-\eta\pi/4}$,  both $\vec{\phi}_0 $ and $U\vec{\phi}_0$ are not parallel to $\vec{\psi}_0.$ Thus  \eqref{non-subord-II} follows from
Lemma \ref{stab-globe-II} and \eqref{U-phi-n}.

Now we show the first equation of \eqref{subord-II}. By  \eqref{symmetry}, we only need to consider $\|\vec{\psi}_n\|.$ We know that
$
\vec{\psi}_n=T_ns.
$
Let $k$ be the unique integer such that
$2^{2k}\le n<2^{2(k+1)}$. We discuss two cases.

{\bf Case I}: $2^{2k}\le n< 2^{2k+1}$. In this case, by \eqref{parlin},  we have
$$
T_n=T_{2^{2k}\to n}T_{2^{2k}}=\overline{T}_{n-2^{2k}}A_{2k}.
$$
Since $n-2^{2k}<2^{2k}$, by \eqref{up-T-k-II},
\begin{equation*}
\|\overline{T}_{n-2^{2k}}\|\le C^ke^{2^k\gamma}\le n^\alpha e^{2^k\gamma}.
\end{equation*}
Combine with \eqref{sub-decay-2n}, we get
\begin{equation}\label{up-globe-1-II}
\|\vec{\psi}_n\|=\|T_ns\|\le \|\overline{T}_{n-2^{2k}}\|\|A_{2k}s\|\lesssim n^\alpha.
\end{equation}

{\bf Case II}: $2^{2k+1}\le n< 2^{2(k+1)}$. In this case, by \eqref{parlin} and \eqref{inverse},  we have
$$
T_n=T_{ n\to 2^{2(k+1)}}^{-1}T_{2^{2(k+1)}}=UT_{2^{2(k+1)}-n}UA_{2(k+1)}.
$$
Since $2^{2(k+1)}-n<2^{2k+1}$, by \eqref{up-T-k-II},
\begin{equation*}
\|T_{2^{2(k+1)}-n}\|\le C^{k+1}e^{2^{k+1}\gamma}\le n^\alpha e^{2^{k+1}\gamma}.
\end{equation*}
Combine with \eqref{sub-decay-2n}, we get
\begin{equation}\label{up-globe-2-II}
\|\vec{\psi}_n\|=\|T_ns\|\le \|T_{2^{2(k+1)}-n}\|\|A_{2(k+1)}s\|\lesssim n^\alpha.
\end{equation}
The first equation of \eqref{subord-II} now follows from \eqref{up-globe-1-II} and \eqref{up-globe-2-II}.
\hfill $\Box$

\section{Type-III energy}\label{sec-type-III}

 In this section, we prove Theorem \ref{main-norm} for type-III energy and Theorem \ref{main-type-III}. Basically we follow the strategy of last section. Since considerable part of the proof is the same, we only sketch it, and focus  on the part which is different.

 Throughout this section, we fix  $E\in \Sigma_{III}$.

\subsection{Exact asymptotic behaviors  of $t_n$, $\mu_n$, $\nu_n$ and $\omega_n$   }\label{5.1}\

\begin{thm}
There exists  $\gamma=\gamma(E)>0$ such that
\begin{equation}\label{t-mu-III}
\begin{cases}
{\displaystyle
\lim_{n\rightarrow\infty}\frac{|t_{2n}(E)|}{e^{2^{n}\gamma}}=\frac{1}{2},\
\lim_{n\rightarrow\infty}\frac{|t_{2n+1}(E)|}{e^{-2^{n}\gamma}}=2,\
\lim_{n\rightarrow\infty}\frac{|\mu_{n}(E)|}{e^{2^{n-1}\gamma}}=\frac{1}{\sqrt{2}},}\\
{\displaystyle\lim_{n\to\infty}\frac{|\nu_{n}(E)|}{e^{2^n
\gamma}}=\frac{\sec\theta(E)}{2}, \ \  \ \ \ \ \ \ \ \
\lim_{n\to\infty}\frac{|\omega_{n}(E)|}{e^{2^n
\gamma}}=\frac{|\tan\theta(E)|}{2}.}
\end{cases}
\end{equation}
where $\theta(E)$ is defined in Lemma \ref{energy-coupling inequality}.
\end{thm}

\begin{proof}
Since $E\in \Sigma_{III}$, by Lemma \ref{char-II-III},
\begin{equation*}\label{0-infty-III}
\lim_{n\rightarrow\infty}|t_{2n+1}(E)|=0\ \ \ \text{ and }\ \ \ \lim_{n\rightarrow\infty}t_{2n}(E)=-\infty.
\end{equation*}
   Write
$$
|t_{2n}(E)|=e^{2^n\gamma_{2n}}\ \ \ \text{ and }\ \ \
|t_{2n+1}(E)|=e^{-2^n\gamma_{2n+1}}.
$$
By repeating  the proof of Theorem \ref{asy-t-mu-II}, we can show that $\lim_n \gamma_n=\gamma>0$ exists.
Moreover the first two equalities of \eqref{t-mu-III} can be obtained in the same way.
The other three equalities follows from
 \eqref{A-B-odd},  \eqref{A-B-even} and \eqref{kappa}.
\end{proof}

We also need the following more precise descriptions of $t_n(E), \nu_n(E)$ and $\omega_n(E)$ to study the structure of $A_n$ and $B_n$.

 \begin{prop}
We have the following expansions:
\begin{equation}\label{t-prod-III}
\begin{array}{rcl}
|t_{2n}(E)|e^{-2^n\gamma}&=&\frac{1}{2}+O( e^{-2^n\gamma})\\
|t_{2n+1}(E)|e^{2^{n}\gamma}&=&2+O( e^{-2^n\gamma})\\
|\nu_n(E)| e^{-2^n\gamma}&=&\frac{\sec \theta(E)}{2}+O( e^{-2^n\gamma})\\
|\omega_n(E)| e^{-2^n\gamma}&=&\frac{|\tan \theta(E)|}{2}+O( e^{-2^n\gamma}).
\end{array}
\end{equation}
\end{prop}

\begin{proof}
  The proof of the first two equations of \eqref{t-prod-III} is similar   as that of Proposition \ref{more-precise-II}.
  By \eqref{A-B-even}, \eqref{kappa} and  \eqref{t-mu-III},  for $n$ large enough, we have
$$
\big||t_{2n}(E)|-|\nu_n(E)|\cos\theta(E)\big|=\frac{|2+t_{2n+1}(E)|}{|t_{2n}(E)|+|\nu_n(E)|\cos\theta(E)}\le 6e^{-2^n\gamma}.
$$
Then the third equation of \eqref{t-prod-III} holds.
By \eqref{kappa},  the last equation of \eqref{t-prod-III} holds.
\end{proof}

\subsection{Structures  of $A_n$ and $B_n$}\label{5.2}\

Write  $\varepsilon_n:={\rm sign}(\nu_n)$. By Lemma \ref{recur-mu-nu-omega} and Lemma \ref{char-II-III},
$\varepsilon_n$ is eventually  constant, which we denote by $\hat \eta$.
By \eqref{kappa}, the sign of $\omega_n(E)$ is equal to the sign of $\nu_n(E)\sin\theta(E)$.
Write
\begin{equation}\label{const-III}
\begin{array}{rl}
C_{\pm}(E):=&I\mp\sec\theta(E)\  V\mp\tan\theta(E)\  W\\
=&\begin{pmatrix}
1\mp\sec\theta(E)&\pm\tan\theta(E)\\
\mp\tan\theta(E)&1\pm\sec\theta(E)
\end{pmatrix}.\end{array}
\end{equation}

\begin{thm}\label{exact-struc-even-III}
  For $n$ large enough, we have
\begin{equation}\label{struc-A-even-III}
\frac{A_{2n}}{t_{2n}}= \frac{C_{\hat \eta}(E)}{2}
+O(e^{-2^{n-1}\gamma}),\quad
\frac{B_{2n}}{t_{2n}}=\frac{C_{-\hat \eta}(E)}{2}+O(e^{-2^{n-1}\gamma}).
\end{equation}
\end{thm}

\begin{proof}
By \eqref{recur-AB-even}, we have
\begin{equation}\label{III-even}
\frac{A_{2n+2}}{t_{2n+2}}=
\frac{t_{2n}t_{2n+1}}{t_{2n+2}}A_{2n}-\frac{t_{2n}\nu_{n}}{t_{2n+2}}V-\frac{t_{2n}
\omega_n}{t_{2n+2}}W+\frac{(1-t_{2n}^2)}{t_{2n+2}}I
\end{equation}
Then by \eqref{t-prod-III},  there exist two constants $c, d>0$ such that
$$\left\|{\frac{A_{2n+2}}{t_{2n+2}}}\right\|\le c+d e^{-2^{n}\gamma}\left\|{\frac{A_{2n}}{t_{2n}}}\right\|.$$
By Lemma \ref{elementary},
$\left\|{{A_{2n}}/{t_{2n}}}\right\|=O(1).$
Together with \eqref{III-even} and \eqref{t-prod-III}, we get
$$
\frac{A_{2n+2}}{t_{2n+2}}=\frac{1}{2}(I-\varepsilon_n\sec\theta\  V-\varepsilon_n\tan\theta\  W)+O(e^{-2^{n}\gamma}).
$$
This proves the first equation of \eqref{struc-A-even-III}.

The proof of the second equation of \eqref{struc-A-even-III} is analogous.
\end{proof}

Since $E\in \Sigma_{III}$, $t_n(E)\ne0$ for all $n\in \N$. Write  $\hat \xi_n=-{\rm sign}(t_{2n-1})$ and $\hat \delta_n=\prod_{j=1}^n\hat \xi_j$.

\begin{thm}\label{exact-struc-odd-III}
For $n$ large enough, we have
 \begin{eqnarray*}\label{struc-AB-odd-III}
t_{2n-1}A_{2n-1}&=& \pm\frac{\hat\delta_n U C_{\hat\eta}(E)}{\sqrt{2}}
+O(e^{-2^{n}\gamma}),\\
t_{2n-1}B_{2n-1}&=&\mp\frac{\hat \delta_n UC_{-\hat\eta}(E)}{\sqrt{2}} +O(e^{-2^{n}\gamma}).
\end{eqnarray*}
\end{thm}

\begin{proof}
By using \eqref{recur-AB-odd}, \eqref{t-mu-III} and  \eqref{t-prod-III},  the same proof as that in Theorem \ref{exact-struc-even}  shows that, there exist two nonzero matrices $C'$ and $\widehat C'$ such that
\begin{equation}\label{A-III}
t_{2n-1}A_{2n-1}= \hat \delta_n C'+O(e^{-2^{n}\gamma});\ t_{2n-1}B_{2n-1}= \hat \delta_n \widehat C'+O(e^{-2^{n}\gamma}).
\end{equation}

To get the exact forms of $C'$ and $\widehat C'$, we need further informations. We postpone the proof to Section \ref{5.3}.
\end{proof}

\begin{rem}\label{osci-III}
{\rm
Recall that $A_{n}=T_{2^{n}}$.
Thus Theorem \ref{exact-struc-even-III} implies that $\|T_k\|\sim e^{\gamma\sqrt{k}} $ for $k=2^{2n}$, while Theorem \ref{exact-struc-odd-III} implies that $\|T_k\|\sim e^{\frac{\gamma}{\sqrt{2}} \sqrt{k}} $ for $k=2^{2n-1}$. This implies that there  exists fluctuation for the coefficients of $\sqrt{k}$. For type-II energy, we observe the same phenomenon, as Remark \ref{osci-II} has suggested.
}
\end{rem}


 \subsection{Proof of Theorem \ref{main-norm} for type-III energy}\label{5.4}\

Similar as Lemma \ref{norm-II}, we have the following result, the proof of which is essentially the same and will be omitted.

\begin{lemma}
There exists constant $C=C(E)>1$ such that
\begin{equation}\label{T-k-III}
C^{-n-2}e^{2^{n-1}\gamma}\le\|T_{ k}\|, \|\overline{T}_k\|\le C^{n+2} e^{2^{n}\gamma}.
\end{equation}
if $2^{2n-1}\le k<2^{2n+1}$.
Consequently, if $k<2^{2n+1},$ then
\begin{equation*}\label{up-T-k-III}
 \|T_{ k}\| , \|\overline{T}_k\|\le C^{n+2} e^{2^{n}\gamma}.
\end{equation*}
\end{lemma}

 \noindent {\bf Proof of Theorem \ref{main-norm}(type-III energy).} \
By Lemma \ref{dp} (iii), we only need to consider $T_n.$ Fix  $n\in \N$, let $k$ be the unique integer such that
$2^{2k-1}\le n<2^{2k+1}$. Then
$
2^{k-1/2}\le \sqrt{n}< 2^{k+1/2}.
$
By \eqref{T-k-III}, we have
\begin{equation*}
 n^{-\alpha} e^{\frac{\gamma}{2\sqrt{2}}\sqrt{n}}\le C^{-k-2}e^{\frac{\gamma}{2\sqrt{2}}\sqrt{n}}
 \le \|T_n\|\le C^{k+2}e^{\sqrt{2}\gamma\sqrt{n}}\le n^{\alpha} e^{\sqrt{2}\gamma\sqrt{n}}.
\end{equation*}
where $\alpha=3\log C/\log 2.$ Then Theorem \ref{main-norm} holds for any $c_2>\sqrt{2}$ and $c_1\in(0,\sqrt{2}/4).$
\hfill $\Box$


\subsection{Stable directions of $\{A_{2n-1}\}$ and $\{B_{2n-1}\}$}\label{5.3}\

By an analogous proof with Proposition \ref{Oseledets-II}, we have the following

\begin{prop}\label{Oseledets-III}
There exist  unit vectors $s, \hat s\in \R^2$ such that
\begin{equation}\label{sub-decay-2n-III}
\|A_{2n-1}s\|,\ \|B_{2n-1}\hat s\|= O(e^{-2^{n-1}\gamma}).
\end{equation}
and for any unit vector $v (\hat v)$ which is independent with $s (\hat s)$,
$$
\|A_{2n-1}v\|,\ \|B_{2n-1}\hat v\| \gtrsim e^{2^{n-1}\gamma}.
$$
\end{prop}

We call $s$ and $\hat s$ the {\it stable directions} of $\{A_{2n-1}\}$ and $\{B_{2n-1}\}$, respectively.
In the following,  we study the exact forms of $s$ and $\hat s$.

\begin{lemma}\label{lem-6.8}
$C_{\hat \eta}(E)s=\vec{0}$  and $ C_{-\hat\eta}(E)\hat s=\vec{0}.$
\end{lemma}

\begin{proof}
At first we show that  $C_{\hat\eta}(E)s=\vec{0}$. Indeed, if otherwise,
$C_{\hat\eta}(E)s\ne\vec{0}.$ By Theorem \ref{exact-struc-even-III}, we have
$
\|A_{2n}s\|=\frac{|t_{2n}|}{2}\|C_{\hat\eta}(E)s+o(1)\|\sim e^{2^n\gamma} $ and $\|B_{2n}\|\sim  e^{2^n\gamma}.
$
Consequently
$$
\|A_{2n+1}s\|=\|B_{2n}A_{2n}s\|\ge \frac{\|A_{2n}s\|}{\|B_{2n}^{-1}\|}=\frac{\|A_{2n}s\|}{\|B_{2n}\|}\sim 1,
$$
which contradicts with \eqref{sub-decay-2n-III}.

The same argument shows  that $C_{-\hat\eta}(E)\hat s=\vec{0}$.
\end{proof}

By \eqref{const-III}, $C_+(E)$ and $C_-(E)$ are of rank one,  and by direct computation,
\begin{equation}\label{C+C-}
C_+(E)v_{-\frac{\theta(E)}{2}}=\vec{0} \ \ \text{ and }\ \ C_-(E)v_{\frac{\theta(E)-\pi}{2}}=\vec{0}.
\end{equation}
Thus we have:

\begin{cor}\label{stab-ustab}
{\rm (i)} If $\hat\eta=+$, then
$$s\parallel v_{-\frac{\theta(E)}{2}},\quad \hat s\parallel v_{\frac{\theta(E)-\pi}{2}}.$$
If $\hat \eta=-$, then
$$s\parallel v_{\frac{\theta(E)-\pi}{2}},\quad \hat s\parallel v_{-\frac{\theta(E)}{2}}.$$
Consequently $Us\parallel \hat s$ and $s\nparallel \hat s.$

{\rm (ii)} If $|E|=\lambda$, then $s\parallel v_0$ or $v_{\pi/2}.$

{\rm (iii)} If $|E|\ne\lambda$, then $s\nparallel  v_0$, $v_{\pm\pi/4}$, $v_{\pi/2}$.
\end{cor}

\begin{proof}
(i) The first two statements follow directly from \eqref{C+C-}.  It is seen that $Uv_{-\frac{\theta(E)}{2}}=v_{\frac{\theta(E)-\pi}{2}},$ consequently $Us\parallel \hat s. $ To have $s\parallel \hat s$, we need $-\frac{\theta(E)}{2}=\frac{\theta(E)-\pi}{2}+k\pi$, or equivalently, $\theta(E)=\frac{\pi}{2}-k\pi$. However by Lemma \ref{energy-coupling inequality}, $|\theta(E)|< \frac{\pi}{2}.$  So $s\nparallel \hat s.$

(ii) By \eqref{kappa} and the definition of $\theta(E)$, $|E|=\lambda$ if and only if $\theta(E)=0$.  Then (ii) follows from (i).

(iii) Since $|E|\ne \lambda$, we have $0<|\theta(E)|<\frac{\pi}{2}.$ The result follows from (i).
\end{proof}

Now we can finish the proof of Theorem \ref{exact-struc-odd-III}:

\noindent {\bf Proof of Theorem \ref{exact-struc-odd-III}(continued).}\
Recall that by \eqref{A-III},
\begin{equation*}
t_{2n-1}A_{2n-1}=\hat\delta_{n} C'+O(e^{-2^{n}\gamma}).
\end{equation*}
By taking the traces on both sides, we get
$
\tr( C' )= \hat\delta_{n}t_{2n-1}^2+O(e^{-2^{n}\gamma}).
$
Since $t_{2n-1}\sim e^{-2^{n-1}\gamma}$, we get $\tr ( C' )=0.$
Moreover, by \eqref{sub-decay-2n-III}, we have
$
A_{2n-1}s= t_{2n-1}^{-1}(\hat \delta_{n}C's+O(e^{-2^{n}\gamma}))=O(e^{-2^{n-1}\gamma}).
$
Since $t_{2n-1}\sim e^{-2^{n-1}\gamma}$, we conclude that
$C's=\vec{0}$.

On the other hand, by Lemma \ref{lem-6.8} and  direct computation, we have
$$
UC_{\hat\eta}(E)s=\vec{0} \ \ \text{ and }\ \ \tr(UC_{\hat\eta}(E))=0.
$$
Thus
 there exists   $c\ne 0$ such that
 $C'=c\ UC_{\hat \eta}(E)$.
By the same argument, there exists  $\hat c\ne 0$ such that
$\hat C'=\hat c\ UC_{-\hat\eta}(E)$.

In the following we compute $c $ and $\hat c.$
By the formulas of $C'$ and $\hat C'$, for $n$ large enough, we have
\begin{eqnarray*}
&&\hat \delta_{n}t_{2n-1}(A_{2n-1}-B_{2n-1})\\
&=&(c-\hat c)U+\hat \eta\sec\theta(c+\hat c)W+\hat\eta\tan\theta(c+\hat c)V+O(e^{-2^{n}\gamma}).
\end{eqnarray*}
Recall that  $A_{2n-1}-B_{2n-1}=\mu_{n}U.$ Then by \eqref{t-mu-III}, we conclude that
$$
t_{2n-1}(A_{2n-1}-B_{2n-1})=\pm(\sqrt{2}+o(1))U.
$$
Consequently,
$$
(\pm\sqrt{2}+\hat c-c)U+\hat\eta\sec\theta(c+\hat c)W+\hat\eta\tan\theta(c+\hat c)V=o(1).
$$
Since $U,V,W$ are linearly independent, we have
$$
\pm\sqrt{2}+\hat c-c=\sec\theta(c+\hat c)=\tan\theta(c+\hat c)=0.
$$
Since $\sec\theta\ne0$, we get $c=-\hat c=\pm\sqrt{2}/2.$
\hfill $\Box$

\smallskip

Next we study the behaviors of $A_{2n}s$ and $B_{2n}\hat s.$

\begin{prop}\label{addition-III}
We have
\begin{equation*}\label{asym-AB-2n-s}
A_{2n}s=\pm \frac{\sqrt{2}}{2}\hat s +O(e^{-2^{n}\gamma}),\quad \text{ and } \ \ \ B_{2n}\hat s=\mp \frac{\sqrt{2}}{2} s +O(e^{-2^{n}\gamma}).
\end{equation*}
\end{prop}
\begin{proof}
By Theorem \ref{exact-struc-odd-III}, without loss  of generality, we take $\hat\eta=+$ and assume
$$
t_{2n-1}A_{2n-1}=\frac{1}{\sqrt{2}}UC_+
+\begin{pmatrix}
\tilde x_n&\tilde u_n\\
\tilde y_n&\tilde z_n
\end{pmatrix}
$$
with $\tilde x_n, \tilde y_n, \tilde u_n, \tilde z_n=O(e^{-2^{n}\gamma})$.
Define
$$c_n=\frac{1}{\sqrt{2}}+\frac{\tilde u_n}{1+\sec\theta}.
$$
Since $\theta=\theta(E)\in (-\pi/2,\pi/2),$ $\sec\theta>0$. Thus $c_n=1/\sqrt{2}+O(e^{-2^{n}\gamma})$. By \eqref{const-III}, there are $x_n,y_n,z_n$ such that
\begin{equation}\label{odd-III}
t_{2n-1}A_{2n-1}=c_n UC_+
+\begin{pmatrix}
x_n&0\\
y_n&z_n
\end{pmatrix}
\end{equation}
with  $x_n, y_n, z_n=O(e^{-2^n\gamma}).$
By \eqref{odd-III} and Remark \ref{rembasic},
\begin{equation}\label{odd-III-B}
t_{2n-1}B_{2n-1}=-c_n UC_-
+\begin{pmatrix}
x_n&-y_n\\
0&z_n
\end{pmatrix}.
\end{equation}
Taking traces on both sides of \eqref{odd-III}, we have
\begin{equation}\label{trace-|||}
t_{2n-1}^2=x_n+z_n.
\end{equation}
Taking determinants  on both sides of \eqref{odd-III}, we have
\begin{equation}\label{det}
t_{2n-1}^2=c_n((x_n-z_n)\tan\theta-y_n(1+\sec\theta))+x_nz_n.
\end{equation}

Since $\hat\eta=+,$ by Corollary \ref{stab-ustab}, without loss of generality, we can take
$$
s=v_{-\frac{\theta}{2}}=(\cos\frac{\theta}{2},\sin\frac{\theta}{2})^t\ \ \text{ and }\ \ \hat s=v_{\frac{\theta-\pi}{2}}=(\sin\frac{\theta}{2},\cos\frac{\theta}{2})^t.
$$
Write
$$
s_{\perp}:=(-\sin\frac{\theta}{2},\cos\frac{\theta}{2})^t\ \ \text{ and }\ \ \hat s_{\perp}:=(\cos\frac{\theta}{2},-\sin\frac{\theta}{2})^t,
$$
 then $s\perp s_\perp$ and $\hat s\perp \hat s_\perp.$  Moreover by direct computation,
\begin{equation}\label{vec}
UC_+=2\sec\theta\  s\cdot s_\perp^t,\quad
UC_-=2\sec\theta\ \hat s\cdot \hat s_\perp^t.
\end{equation}
Thus $(UC_+)^2=(UC_-)^2$ is zero matrix.
Write
$$
X:=\begin{pmatrix}
x_n&0\\
y_n&z_n
\end{pmatrix},\quad Y:=\begin{pmatrix}
x_n&-y_n\\
0&z_n
\end{pmatrix}
$$
and
$$\begin{array}{rccl}
\Delta_1:=&\hat s_\perp^tX  s&=&
x_n\cos^2\frac{\theta}{2}-y_n\sin\frac{\theta}{2}\cos\frac{\theta}{2}-z_n\sin^2\frac{\theta}{2}\\
\Delta_2:=&s_\perp^t  Y\hat s&=&
-x_n\sin^2\frac{\theta}{2}+y_n\sin\frac{\theta}{2}\cos\frac{\theta}{2}+z_n\cos^2\frac{\theta}{2}\\
\Delta_3:=&\hat s_\perp^tYX s&=&
x_n^2\cos^2\frac{\theta}{2}-(y_n\cos\frac{\theta}{2}+z_n\sin\frac{\theta}{2})^2\\
\Delta_4:=&s_\perp^t XY\hat s&=&
z_n^2\cos^2\frac{\theta}{2}-(y_n\cos\frac{\theta}{2}-x_n\sin\frac{\theta}{2})^2.
\end{array}
$$
By \eqref{odd-III}, \eqref{odd-III-B}, \eqref{vec} and Lemma \ref{lem-6.8}, we have
\begin{equation*}
\begin{array}{l}
t_{2n-1}^2A_{2n}s=-c_nUC_- Xs+YXs=-(2c_n\Delta_1\sec\theta)\hat s
+O(e^{-2^{n+1}\gamma}),\\
t_{2n-1}^2B_{2n}\hat s=c_nUC_+Y\hat s+XY\hat s=
(2c_n\Delta_2\sec\theta) s
+O(e^{-2^{n+1}\gamma}).
\end{array}
\end{equation*}
It is seen that $\Delta_1,\Delta_2=O(e^{-2^n\gamma}).$ Since $c_n=1/\sqrt{2}+O(e^{-2^{n}\gamma}), $ we have
\begin{equation}\label{fab}
\begin{array}{l}
 A_{2n}s=-\sqrt{2}\sec\theta \ t_{2n-1}^{-2}\Delta_1 \ \hat s
+O(e^{-2^{n}\gamma}),\\
 B_{2n}\hat s=\ \ \sqrt{2}\sec\theta\ t_{2n-1}^{-2}\Delta_2\   s
+O(e^{-2^{n}\gamma}).
\end{array}
\end{equation}

Now we study $t_{2n-1}^{-2}\Delta_i$. Since $C_+s=\vec{0}$ and $(UC_-)^2$ is zero matrix, by \eqref{vec},
$$
\begin{array}{cl}
&t_{2n-1}^4A_{2n+1}s\\
=&(t_{2n-1}A_{2n-1})(t_{2n-1}B_{2n-1})(t_{2n-1}B_{2n-1})(t_{2n-1}A_{2n-1})s\\
=&(c_nUC_++X)(-c_nUC_-+Y)(-c_nUC_-+Y)(c_nUC_++X)s\\
=&-c_n^2UC_+UC_-YXs-c_n^2UC_+YUC_-Xs+O(e^{-3\times2^n\gamma})\\
=&-(4c_n^2\sec^2\theta)s
\left( s_\perp^t \hat s \hat s_\perp^t YXs+s_\perp^t Y \hat s \hat s_\perp^t X s
\right)+O(e^{-3\times2^n\gamma})\\
=&-4c_n^2\sec^2\theta(\Delta_3\cos\theta+\Delta_1\Delta_2)s
+O(e^{-3\times2^n\gamma}).
\end{array}
$$
By  \eqref{sub-decay-2n-III}, $A_{2n+1}s=O(e^{-2^n\gamma})$.  By \eqref{t-mu-III}, $t_{2n-1}^4\sim O(e^{-2^{n+1}\gamma})$. Thus we have
\begin{equation*}
\Delta_3\cos\theta+\Delta_1\Delta_2=O(e^{-3\times2^n\gamma}).
\end{equation*}
Analogously, we have
$$
t_{2n-1}^4B_{2n+1}\hat s=
-4c_n^2\sec^2\theta(\Delta_4\cos\theta+\Delta_1\Delta_2)\hat s
+O(e^{-3\times2^n\gamma}),
$$
and consequently
\begin{equation*}
\Delta_4\cos\theta+\Delta_1\Delta_2=O(e^{-3\times2^n\gamma}).
\end{equation*}
Since $\theta\in(-\pi/2,\pi/2),$ we have $\cos \theta\ne 0$.  Thus we conclude that
\begin{equation*}
\Delta_3-\Delta_4=O(e^{-3\times2^n\gamma}).
\end{equation*}
On the other hand, by \eqref{trace-|||},
$$
\Delta_3-\Delta_4=(x_n+z_n)(x_n-z_n-y_n\sin\theta)=t_{2n-1}^2(x_n-z_n-y_n\sin\theta),
$$
Thus we have
$$
t_{2n-1}^{-2}(x_n-z_n-y_n\sin\theta)=t_{2n-1}^{-4}(\Delta_3-\Delta_4)=O(e^{-2^n\gamma}).
$$
Together with \eqref{trace-|||}, \eqref{det}, we get
$$
\begin{cases}
t_{2n-1}^{-2}x_n&=\frac{1}{2}-\frac{\sqrt{2}\sin\theta}{2(1+\cos\theta)}+O(e^{-2^{n}\gamma}),\\
t_{2n-1}^{-2}y_n&=-\frac{\sqrt{2}}{1+\cos\theta}+O(e^{-2^{n}\gamma}),\\
t_{2n-1}^{-2}z_n&=\frac{1}{2}+\frac{\sqrt{2}\sin\theta}{2(1+\cos\theta)}+O(e^{-2^{n}\gamma}).\\
\end{cases}
$$
Then $t_{2n-1}^{-2}\Delta_1, t_{2n-1}^{-2}\Delta_2=\frac{\cos\theta}{2}+O(e^{-2^{n}\gamma})$.  By \eqref{fab}, the result follows.
\end{proof}


\subsection{One-sided Subordinate solution}\label{5.5}\

In this subsection, we prove Theorem \ref{main-type-III}.
At first, analogous to Lemma \ref{stab-globe-II}, we have:

\begin{lemma}\label{stab-globe-III}
There exist $\alpha=\alpha(E)>0$ and  unit vector $\hat v$ such that for any  $n\in\N$ and  any vector $w\nparallel \hat v$, we have
$$
\|T_n w\|\gtrsim n^{-\alpha} e^{\frac{\gamma}{2\sqrt{2}}\sqrt{n}}.
$$
\end{lemma}

\smallskip

Recall that $\psi^\theta$ is the solution of $H_\lambda\psi=E\psi$ with $\vec{\psi}_0=v_\theta.$

\noindent {\bf Proof of Theorem \ref{main-type-III}.\ }
(i) Let  $s$ and $\hat s$ be the   the stable directions of $\{A_{2n-1}:n\ge 1\}$ and $\{B_{2n-1}:n\ge 1\}$, respectively.
Assume $\hat \theta$ is such that $v_{\hat \theta}=s.$
 By Corollary \ref{stab-ustab}, $\hat \theta\ne \pm\pi/4+k\pi$ and    $\hat s\parallel v_{-\hat\theta -\frac{\pi}{2}}$.
Define  $\psi^r:=\psi^{\hat \theta}$ and $\psi^l:=\psi^{-\hat \theta-\frac{\pi}{2}}.$ Let us check that $\psi^r$ and $\psi^l$ satisfy the desired properties.

We deal with $\psi^r$ firstly. At first we claim that $\hat v\parallel s,$ where $\hat v$ is from Lemma \ref{stab-globe-III}.
Indeed if on the contrary,  we have
$$
\|T_ns\|\gtrsim n^{-\alpha}e^{\frac{\gamma}{2\sqrt{2}}\sqrt{n}}\to \infty.
$$
On the other hand, by Proposition \ref{Oseledets-III}, we have
$$
\|T_{2^{2n-1}}s\|=\|A_{2n-1}s\|=O(e^{-2^{n-1}\gamma})\to 0.
$$
 Thus we get a contradiction.

 Since $\vec{\psi}_n^r=T_n\vec{\psi}_0^r=T_ns,$ the third equation of \eqref{subord-r-III} follows from Proposition \ref{Oseledets-III}; since $\hat s\parallel v_{-\hat \theta-\frac{\pi}{2}}$,  the fourth  equation of \eqref{subord-r-III} follows from Proposition \ref{addition-III}.

By Corollary \ref{stab-ustab} (i), $s\nparallel \hat s$. Since $Uv_{\hat \theta}=v_{-\frac{\pi}{2}-\hat \theta}$, by \eqref{formu-reverse} and Lemma \ref{stab-globe-III},
$$
\|\vec{\psi}^r_{-n}\| =
\|T_{-n}\vec{\psi}^{\hat \theta}_0\|=\|UT_{n}Uv_{\hat \theta}\|=
\|T_{n}v_{-\frac{\pi}{2}-\hat \theta}\|=\|T_n\hat s\|\gtrsim e^{\frac{\gamma}{4}\sqrt{n}}.
$$
That is, the second equation of \eqref{subord-r-III} holds.

Finally the proof of the first equation of \eqref{subord-r-III} is the same as that of the first equation of \eqref{subord-II}.

Now we consider $\psi^l.$ By \eqref{formu-reverse}, for $n\in\Z$,
$$
\vec{\psi}^l_{-n}=
T_{-n}v_{-\frac{\pi}{2}-\hat \theta}=UT_{n}Uv_{-\frac{\pi}{2}-\hat \theta}=
UT_{n}v_{\hat \theta}=U\vec{\psi}_n,
$$
or equivalently, $\psi^l_n=\psi^r_{1-n}$ for $n\in \Z$. Then \eqref{subord-l-III} follows from \eqref{subord-r-III}.

(ii)    By the assumption, $\vec{\phi}_0 \nparallel s, \hat s.$ Then for $n\in\N$, by Lemma \ref{stab-globe-III},
$$
\|\vec{\phi}_n\|=\|T_n\vec{\phi}_0\|\gtrsim e^{\frac{\gamma}{4}\sqrt{n}}.
$$
Since $\vec{\phi}_0\nparallel \hat s$, by Corollary \ref{stab-ustab} (i), $U\vec{\phi}_0 \nparallel U \hat s\parallel s.$ By \eqref{formu-reverse} and Lemma \ref{stab-globe-III},
$$
\|\vec{\phi}_{-n}\|=
\|T_{-n}\vec{\phi}_0\|=\|UT_{n}U\vec{\phi}_0\|=\|T_{n}U\vec{\phi}_0\|
 \gtrsim e^{\frac{\gamma}{4}\sqrt{n}}.
$$
Thus \eqref{non-subord-III} holds.
 \hfill $\Box$

\section{Local dimensions of the  spectral measure}\label{sec-loc-dim}

In this section, we  prove Theorem \ref{main-locdim}. The main tool is the subordinacy theory developed in \cite{GP,G,KP,JL1,JL00}.
Let us recall the basic settings and the main results of \cite{JL1,JL00}.

\subsection{Basic notations and  facts}\

Fix a bounded real  potential $V$ and consider the operator $H=H_V.$
Let $\Z^+=\{1,2,3,\cdots\}$ and $\Z^-=\{\cdots,-2,-1,0\}$.  Consider the family of phase boundary conditions:
\begin{equation}\label{bdy-condition}
\psi_0\cos\beta+\psi_1\sin\beta=0,
\end{equation}
where $-\pi/2< \beta\le\pi/2.$

 Let $z=E+i\epsilon$ and consider equation
 \begin{equation}\label{eigen-equation}
 Hu=zu.
 \end{equation}
 It is known that for $\epsilon>0$, the eigen-equation  has unique(up to normalization) solutions $\hat u^{\pm}_z$ 
 which are in  $\ell^2$ at, correspondingly, $\pm\infty.$ Let $\hat u^{\pm}_{\beta,z}$ be $\hat u^{\pm}_z$ normalized by 
 $\hat u^{\pm}_{\beta,z}(0)\cos\beta+\hat u^{\pm}_{\beta,z}(1)\sin\beta=1.$  We denote by $u^{\pm}_{\beta,z}$
 the solution of \eqref{eigen-equation} on $\Z^{\pm}$  satisfying the boundary condition \eqref{bdy-condition} and
 normalized by $|u_{\beta,z}^\pm(0)|^2+|u_{\beta,z}^\pm(1)|^2=1.$
 For $z$ in the upper half-plane, the right and left Weyl-Titchmarsh $m$-functions are uniquely defined by
 \begin{equation*}
\hat u_{\beta,z}^\pm=u_{\beta+\pi/2,z}^\pm\mp m_\beta^\pm(z)u_{\beta,z}^\pm.
\end{equation*}
From the definition, we know that
$
m_{\beta+\pi}^\pm(z)=m_\beta^\pm(z).
$

For $\Im z>0$,  define the whole-line $m$-function as
\begin{equation*}
M(z)=\frac{m_0^+(z)m_0^-(z)-1}{m_0^+(z)+m_0^-(z)}=
\frac{m_\beta^+(z)m_\beta^-(z)-1}{m_\beta^+(z)+m_\beta^-(z)}.
\end{equation*}
$M(z)$ coincides with the Borel transform of the spectral measure $\mu$ of $H$:
\begin{equation*}
M(z)=\int\frac{d\mu(x)}{x-z}.
\end{equation*}
It is shown in  \cite{dJLS} that, for $\alpha\in [0,1),$
\begin{equation}\label{M-density}
\begin{cases}
\limsup_{\epsilon\to 0} \epsilon^{1-\alpha}|M(E+i\epsilon)|=\infty&\Leftrightarrow D_\mu^\alpha(E)=\infty\\
\limsup_{\epsilon\to 0} \epsilon^{1-\alpha}|M(E+i\epsilon)|=0&\Leftrightarrow D_\mu^\alpha(E)=0.
\end{cases}
\end{equation}

For any $u:\Z^+\to\R$,  define the norm of $u$ over an interval of length $L$ as
\begin{equation*}
\|u\|_L:=\left( \sum_{n=1}^{[L]} |u(n)|^2+ (L-[L]) |u([L]+1)|^2\right)^{1/2},
\end{equation*}
where $[L]$ denotes the integer part of $L.$ Similarly,
 for any $u:\Z^-\to\R$,  define the norm of $u$ over an interval of length $L$ as
\begin{equation*}
\|u\|_L:=\left( \sum_{n=0}^{[L]-1} |u(-n)|^2+ (L-[L]) |u(-[L])|^2\right)^{1/2}.
\end{equation*}
Now given any $\epsilon>0$, define the lengths $L_\beta^\pm(\epsilon)$ by requiring the equality
\begin{equation}\label{def-L-epsilon}
\|u_{\beta,E}^\pm\|_{L_\beta^\pm(\epsilon)}\|u_{\beta+\frac{\pi}{2},E}^\pm\|_{L_\beta^\pm(\epsilon)}=\frac{1}{2\epsilon}.
\end{equation}
Note that $L_\beta^\pm(\epsilon)\to \infty$ when $\epsilon\to 0.$
It is proven in \cite{JL1} that
\begin{equation}\label{est-M+}
|m_\beta^+(E+i\epsilon)|\sim
\frac{\|u_{\beta+\frac{\pi}{2},E}^+\|_{L_\beta^+(\epsilon)}}{\|u_{\beta,E}^+\|_{L_\beta^+(\epsilon)}}.
\end{equation}

Now introduce  operator $\mathscr{U}:\ell^2(\Z)\to \ell^2(\Z)$ as  $(\mathscr{U}\psi)_n=\psi_{1-n},  n\in\Z.$
By this operator, the space $\ell^2(\Z^-)$ is naturally identified with $\ell^2(\Z^+)$.
If we define an operator $\tilde H:=\mathscr{U}H\mathscr{U}^{-1}$,
it is seen that $\tilde H$ is another Schr\"odinger  operator with  potential $\tilde V$,
 which is the reflection of $V$ with respect to the point $1/2.$
 Let $\tilde m^+_\beta, \tilde u_{\beta}^\pm, \tilde L_\beta^\pm$, denote, correspondingly,
 $ m^+_\beta,  u_{\beta}^\pm,  L_\beta^\pm$ of the operator $\tilde H.$ Then by an elementary computation, we have
\begin{equation*}
M(z)=\frac{m_\beta^+(z)\tilde m_{\frac{\pi}{2}-\beta}^+(z)-1}{m_\beta^+(z)+\tilde m_{\frac{\pi}{2}-\beta}^+(z)}.
\end{equation*}

Now we go back to the  Thue-Morse potential  $\lambda w$.
By \eqref{symmetry-TM}, we know that $\lambda w$ is symmetric with respect to $1/2$,
consequently we have $\tilde H_\lambda=H_\lambda$ and hence  $\tilde m_\beta^+=m_\beta^+.$
Thus for Thue-Morse Hamiltonian $H_\lambda$, the Weyl's $m$-function satisfies
\begin{equation}\label{formula-M}
M(z)=\frac{m_\beta^+(z) m_{\frac{\pi}{2}-\beta}^+(z)-1}{m_\beta^+(z)+m_{\frac{\pi}{2}-\beta}^+(z)}, \ \ \ \
(\forall \beta\in (-\pi/2,\pi/2]).
\end{equation}

\subsection{Proof of Theorem \ref{main-locdim}}\

At first we consider   type-I energy.  By  \eqref{bd-transfer},
if $E\in \Sigma_I$, then for any $\beta\in (-\pi/2,\pi/2]$,
\begin{equation}\label{norm-u-L-I}
\|u_{\beta,E}^+\|_L\sim L^{1/2}.
\end{equation}

\noindent{\bf Proof of Theorem \ref{main-locdim} ($E\in \Sigma_I$).}  \
Take $\beta=\pi/4.$ By \eqref{norm-u-L-I} and \eqref{est-M+},
we have $|m_{\pi/4}^+(E+i\epsilon)|\sim 1.$ So by \eqref{formula-M},
$
|M(E+i\epsilon)|\lesssim 1.
$
Then by \eqref{M-density}, $D_\mu^\alpha(E)=0$ for all $\alpha<1$. Thus by \eqref{cha-locdim}, we have
$d_\mu(E)\ge 1. $

Next we show that $\mu(\Sigma_I)=0.$ If otherwise, we have $\mu(\Sigma_I)>0.$ Since $d_\mu(E)\ge 1$ for every $E\in \Sigma_I$, it is well-known that $\dim_H \Sigma_I\ge 1$(see for example \cite{Fa}). On the other hand,  notice that  $\Sigma_I$ is countable, then $\dim_H\Sigma_I=0$, we get a contradiction.
\hfill $\Box$

\smallskip

Next we consider  type-II energy.
\begin{lemma}\label{box-norm}
If $E\in \Sigma_{II}$, then there exists $\alpha,\gamma>0$  and $\theta=\pm\pi/4$ such that
\begin{equation}\label{norm-u-L-II}
\|u_{\theta,E}^+\|_L\lesssim L^{\alpha+1} \ \ \text{and }\ \
\|u_{\beta,E}^+\|_L\gtrsim e^{{\gamma\sqrt{L}}/{4}},\ \ (\forall \beta\in (-\pi/2,\pi/2]\setminus\{\theta\}).
\end{equation}
\end{lemma}

\begin{proof}
Let $\psi$ be the unique subordinate solution in Theorem \ref{main-type-II}, assume $\vec{\psi}_0=v_{\theta}$, then $\theta=\pm\pi/4.$ By the definition of $u_{\theta,E}^+$ and the first inequality of \eqref{subord-II},   for any $n\ge 1,$
we have $
|u_{\theta,E}^+(n)|=|\psi(n)|\lesssim
 n^{\alpha}.
$
Consequently
\begin{equation*}
\|u_{\theta,E}^+\|_L\le \left(\sum_{n=1}^{[L]+1} |u_{\theta,E}^+(n)|^2\right)^{1/2}\lesssim L^{\alpha+1}.
\end{equation*}

 If $\beta\in (-\pi/2,\pi/2]$ and $\beta\ne \theta$, by \eqref{non-subord-II} we have
$
|u_{\beta,E}^+(n)|^2+|u_{\beta,E}^+(n+1)|^2  \gtrsim e^{\gamma \sqrt{n}/2}.
$
In particular,
$
\|u_{\beta,E}^+\|_L^2\ge|u_{\beta,E}^+([L]-1)|^2+|u_{\beta,E}^+([L])|^2 \gtrsim e^{{\gamma\sqrt{L}}/{2}}.
$
\end{proof}

\noindent{\bf Proof of Theorem \ref{main-locdim} ($E\in \Sigma_{II}$).}  Fix any $E\in\Sigma_{II}.$ Assume $\theta=\pm\pi/4$ is given by Lemma \ref{box-norm},
then by \eqref{est-M+} and  \eqref{norm-u-L-II},
\begin{equation}\label{7.10}
|m_{\theta}^+(E+i\epsilon)|\sim\frac{\|u_{\theta+\pi/2,E}^+
\|_{L_\theta^+(\epsilon)}}{\|u_{\theta,E}^+\|_{L_\theta^+(\epsilon)}}
\gtrsim \frac{e^{\gamma\sqrt{L_\theta^+(\epsilon)}/4}}{(L_{\theta}^+(\epsilon))^{\alpha+1}}
\gtrsim e^{\gamma\sqrt{L_\theta^+(\epsilon)}/5}.
\end{equation}
Since $\theta=\pm \pi/4$, we have  $m^+_{\theta}=m^+_{\pi/2-\theta}$.  By \eqref{formula-M} and \eqref{7.10},
\begin{equation*}
|M(E+i\epsilon)|\sim|m_{\theta}^+(E+i\epsilon)|.
\end{equation*}
Fix any $\eta\in (0,1)$, by \eqref{def-L-epsilon}, \eqref{est-M+} and \eqref{norm-u-L-II}, we have
\begin{eqnarray*}
\epsilon^{1-\eta}|M(E+i\epsilon)|&\sim&\epsilon^{1-\eta}|m_\theta^+(E+i\epsilon)|
\sim\frac{\|u_{\theta+\pi/2,E}^+\|_{L_\theta^+(\epsilon)}^\eta}{\|u_{\theta,E}^+\|_{L_\theta^+(\epsilon)}^{2-\eta}}\\
&\gtrsim& \frac{e^{\gamma\eta\sqrt{L_\theta^+(\epsilon)}/4}}{(L_{\theta}^+(\epsilon))^{(\alpha+1)(2-\eta)}}
\gtrsim e^{\gamma\eta\sqrt{L_\theta^+(\epsilon)}/5}\to \infty
\end{eqnarray*}
when $\epsilon\downarrow 0.$ By \eqref{M-density}, we conclude that $D_\mu^\eta(E)=\infty.$
Then by \eqref{cha-locdim}, we conclude that $d_\mu(E)=0.$

Now a well known result in fractal geometry shows that $\dim_H \Sigma_{II}=0$ (see for example \cite{Fa}).
\hfill $\Box$

\smallskip

At last,  we consider  type-III energy.

\begin{lemma}
If $E\in \Sigma_{III},$ then there exist $\alpha,\gamma>0$, $\beta(E)\in (-\pi/2,\pi/2]$ such that
\begin{equation}\label{norm-u-L-III}
\|u_{\beta(E),E}^+\|_L\lesssim L^{\alpha+1} \ \ \text{and }\ \ \|u_{\beta,E}^+\|_L\gtrsim e^{{\gamma\sqrt{L}}/{4}},\ \ (\forall \beta\in (-\pi/2,\pi/2]\setminus\{\beta(E)\}).
\end{equation}
Moreover if $E\in \Sigma_{III,1}$, then $\beta(E)=0$ or $\pi/2$; if $E\in\Sigma_{III,2}$, then $\beta(E)\ne 0,\pm\pi/4,\pi/2.$
\end{lemma}

\begin{proof}
By the proof of Theorem \ref{main-type-III}, $\psi^r$ is the solution with $\vec{\psi}^r_0=s$, where $s$ is the stable direction of $\{A_{2n-1}:n\in\N\}$. Assume $\beta(E)\in (-\pi/2,\pi/2]$ is such that $s=v_{\beta(E)}$, then $u_{\beta(E),E}^+(n)=\psi^r(n)$ for $n\ge 1.$
Now the proof of \eqref{norm-u-L-III}  is the same as that of \eqref{norm-u-L-II}, by using \eqref{subord-r-III}, \eqref{subord-l-III} and \eqref{non-subord-III}.
  $\beta(E)=0$ or $\pi/2$ for $E\in \Sigma_{III,1}$ follows from the definition.
$\beta(E)\ne 0,\pm\pi/4,\pi/2$ for $E\in \Sigma_{III,2}$ follows
from the definition and Corollary \ref{stab-ustab} (iii).
\end{proof}

  Indeed,  we know more about $\Sigma_{III,1}$:

\begin{lemma}\label{sigma-3-1}
$\Sigma_{III,1}\subset \{\lambda, -\lambda\}$. Moreover,
 if $|\lambda|>1,$ then  $\Sigma_{III,1}=\emptyset$; on the other hand, there exists a uncountable set $\Gamma\subset [-1,1]\setminus\{0\}$ such that for each $\lambda\in \Gamma$, we have $\Sigma_{III,1}= \{\lambda, -\lambda\}.$
\end{lemma}

\begin{proof}
At first we show that $\Sigma_{III,1}\subset \{\lambda, -\lambda\}$.
Assume $\Sigma_{III,1}$ is nonempty and take $E\in \Sigma_{III,1}.$ By the definition of $\Sigma_{III,1}$,   $\vec{\psi}^r_0=v_0$ or $v_{\frac{\pi}{2}}$. By the proof of Theorem \ref{main-type-III},   $\vec{\psi}^r_0=s$, where $s$ is the stable direction of $\{A_{2n-1}:n\in\N\}$. By   Corollary \ref{stab-ustab},  $E=\pm \lambda.$

Now we show that $\Sigma_{III,1}$ is empty when $|\lambda|>1$. If otherwise, take $E\in \Sigma_{III,1}$, then $E=\pm \lambda.$ By \eqref{t12}, $t_1(E)=-2$ and $t_2(E)=2-4\lambda^2<-2.$ Then by \eqref{recurrence},
$
t_3(E)=2-16\lambda^2<-2.
$
Thus $E\not\in \sigma_2\cup\sigma_3.$ Hence  $E\not\in \sigma(H_\lambda) $ by \eqref{structure-spectrum}, which is a contradiction.

Define a curve $\mathcal C$ as
$$
\mathcal C:=\{f(-2,y): \frac{5}{4}\le y\le \frac{7}{4}\}.
$$
Since $f(-2,y)=(4y-6,4y^2(y-2)+2)$, it is easy to check  that  $\mathcal C\subset S$,
$\mathcal C$ stays below  the parabola $\mathscr{P}$ and touches  the left and right boundaries of $S$. Then the same proof as  Theorem \ref{exist-dense} shows that $(\Lambda_\omega\cap \mathcal C)\setminus \mathscr{P}\ne \emptyset$ for each $\omega\in \{0,1\}^\N$. Assume $y_\omega \in [5/4,7/4]$ is such that $f(-2,y_\omega)\in (\Lambda_\omega\cap \mathcal C)\setminus \mathscr{P}$, then $y_\omega\ne y_{\tilde \omega}$ for $\omega\ne\tilde \omega.$ For each $\omega$, define
$$
\lambda_{\omega,\pm}:= \pm \frac{\sqrt{2-y_\omega}}{2}.
$$
Define $\Gamma:=\{\lambda_{\omega,\pm}:\omega\in \{0,1\}^\N\}$.

Now fix $\lambda\in \Gamma$, we will show that $\pm\lambda\in \sigma(H_\lambda)$ and they are type-III energies.  At first, by \eqref{t12},  we have $t_1(\pm\lambda)=-2$ and $(t_2(\pm\lambda))=2-4\lambda^2$. By the definition of $\Gamma$, there exists $y_\omega\in [5/4,7/4]$ such that
$$
f(t_1(\pm\lambda),t_2(\pm\lambda))=f(-2,2-4\lambda^2)=f(-2,y_\omega)\in \Lambda_\omega\setminus \mathscr{P}.
$$
On the other hand, we know that
$$
f(t_1(\pm\lambda),t_2(\pm\lambda))=(t_3(\pm\lambda),t_4(\pm\lambda))=\varphi_3(\pm\lambda).
$$
This means that $\varphi_3(\pm\lambda)\in \Lambda(S)\setminus \mathscr{P}$. Hence by Lemma \ref{char-II-III}, we conclude that $\pm\lambda\in \Sigma_{III}.$ Now by Corollary \ref{stab-ustab}, $\Sigma_{III,1}=\{\lambda,-\lambda\}$.
\end{proof}

\noindent{\bf Proof of Theorem \ref{main-locdim} ($E\in \Sigma_{III}$).}  \
At first we fix $E\in\Sigma_{III,1}$. Then we have $\beta(E)=0$ or $\pi/2$.
 By \eqref{norm-u-L-III} and \eqref{est-M+}, we have
 \begin{eqnarray*}
  |m_{\beta(E)}^+(E+i\epsilon)|&\sim&
\frac{\|u_{\beta(E)+\pi/2,E}^+\|_{L_{\beta(E)}^+(\epsilon)}}{\|u_{\beta(E),E}^+\|_{L_{\beta(E)}^+(\epsilon)}}\to \infty\\
\nonumber |m_{\frac{\pi}{2}-\beta(E)}^+(E+i\epsilon)|&\sim&
\frac{\|u_{\pi-\beta(E),E}^+\|_{L_{\beta(E)}^+(\epsilon)}}{\|u_{\pi/2-\beta(E),E}^+\|_{L_{\beta(E)}^+(\epsilon)}}=
\frac{\|u_{\beta(E),E}^+\|_{L_{\beta(E)}^+(\epsilon)}}{\|u_{\beta(E)+\pi/2,E}^+\|_{L_{\beta(E)}^+(\epsilon)}}\to 0
 \end{eqnarray*}
 when $\epsilon\to 0$. Consequently $|m_{\beta(E)}^+(E+i\epsilon)||m_{\pi/2-\beta(E)}^+(E+i\epsilon)|\sim1$.
 Then by \eqref{formula-M}, we know that
\begin{equation*}
|M(E+i\epsilon)|\lesssim |m_{\beta(E)}^+(E+i\epsilon)|^{-1}.
\end{equation*}
We know that ( equation (3.1) of \cite{dJLS} )
$$
\Im M(E+i\epsilon)=\epsilon\int_{-\infty}^\infty\frac{d\mu(y)}{(E-y)^2+\epsilon^2}\ge\frac{\mu(E-\epsilon,E+\epsilon)}{2\epsilon}.
$$
Thus for any $\eta\in (0,2),$ by \eqref{est-M+}, \eqref{def-L-epsilon} and \eqref{norm-u-L-III},
\begin{eqnarray*}
\frac{\mu((E-\epsilon,E+\epsilon))}{(2\epsilon)^\eta}&\lesssim&
\epsilon^{1-\eta}|M(E+i\epsilon)|\lesssim \epsilon^{1-\eta} |m_{\beta(E)}^+(E+i\epsilon)|^{-1}\\
&\sim&{\|u_{\beta(E)+\pi/2,E}^+\|_{L_{\beta(E)}^+(\epsilon)}^{\eta-2}}{\|u_{\beta(E),E}^+\|_{L_{\beta(E)}^+(\epsilon)}^\eta}\\
&\lesssim& e^{\gamma(\eta-2)\sqrt{L_{\beta(E)}^+(\epsilon)}/4}L_{\beta(E)}^+(\epsilon)^{\eta(\alpha+1)}\to0.
\end{eqnarray*}
Thus  $D_\mu^\eta(E)=0$ for all $\eta<2$. Thus by \eqref{cha-locdim}, we have
$d_\mu(E)\ge 2. $

Since $d_\mu(E)\ge2$ for any $E\in \Sigma_{III,1}$,  and $\Sigma_{III,1}$ contains at most two points by Lemma \ref{sigma-3-1},  the same argument for  $\mu(\Sigma_I)=0$ shows that $\mu(\Sigma_{III,1})=0$.

Now fix $E\in\Sigma_{III,2}$, then  $\beta(E)\ne 0, \pm\pi/4, \pi/2$.
Still by \eqref{norm-u-L-III} and \eqref{est-M+},
\begin{equation}\label{713}
|m_{\beta(E)}^+(E+i\epsilon)|\sim
\frac{\|u_{\beta(E)+\pi/2,E}^+\|_{L_{\beta(E)}^+(\epsilon)}}{\|u_{\beta(E),E}^+\|_{L_{\theta(E)}^+(\epsilon)}}
\to \infty.
\end{equation}
On the other hand, since $\beta(E)\ne 0,\pm\pi/4, \pi/2,$ we conclude that
$$
\pi-\beta(E), \ \ \pi/2-\beta(E)\not\equiv \beta(E)\pmod \pi.
$$
Consequently $u_{\pi-\beta(E),E}^+$ and $u_{\frac{\pi}{2}-\beta(E),E}^+$ are independent with $u_{\beta(E),E}^+.$
Notice that the space of solutions of $H_\lambda u=Eu$ has dimension $2$,
and
$$\{u_{\beta(E),E}^+, u_{\beta(E)+\pi/2,E}^+\}
$$ is a basis of that space,
then there exist real numbers  $s_1,s_2$ and $t_1,t_2$, with $s_2,t_2\ne 0,$ such that
\begin{eqnarray*}
u_{\pi-\beta(E),E}^+&=&s_1u_{\beta(E),E}^++s_2u_{\beta(E)+\pi/2,E}^+,\\
u_{\frac{\pi}{2}-\beta(E),E}^+&=&t_1u_{\beta(E),E}^++t_2u_{\beta(E)+\pi/2,E}^+.
\end{eqnarray*}
By \eqref{norm-u-L-III}, we get
\begin{equation*}
\|u_{\pi-\beta(E),E}^+\|_{L}, \ \ \|u_{\frac{\pi}{2}-\beta(E),E}^+\|_{L}\sim \|u_{\beta(E)+\pi/2,E}^+\|_L
\end{equation*}
By \eqref{est-M+}, we get
\begin{equation*}
|m_{\pi/2-\beta(E)}^+(E+i\epsilon)|\sim 1.
\end{equation*}
Combine with  \eqref{formula-M} and \eqref{713}, we conclude that
\begin{equation*}
|M(E+i\epsilon)|\sim 1.
\end{equation*}
Then by \eqref{M-density}, $D_\mu^\alpha(E)=0$ for all $\alpha<1$. By \eqref{cha-locdim}, we have
$d_\mu(E)\ge 1. $
\hfill $\Box$


\section{Appendix}

 In this appendix, we give another proof for the fact  that the Thue-Morse Hamiltonian has no point spectrum.

\begin{thm}{\rm (\cite{HKS,DGR})}\label{no-point}
$H_\lambda$ has no point spectrum.
\end{thm}

We begin with the following observation:

\begin{lemma}
If $E\in \R$ is such that $H_\lambda\phi=E\phi$ has a solution $\phi$ with $\phi_{\pm n}\to 0$ when $n\to\infty,$
then   $E$ is a type-II energy.
\end{lemma}

\proof\ Since  the two-sided Thue-Morse sequence is symmetric w.r.t. $1/2$,
$\vec{\phi}_0$ has only two choices: $v_{\pi/4}$ or $v_{-\pi/4}$, as argued in \cite{DGR}.
Without loss of generality, we assume $\vec{\phi}_0=v_{\pi/4}=(1,-1)^t/\sqrt{2}$. Then we have
$$
A_{2n}\vec{\phi}_0= (\phi_{2^{2n}+1},\phi_{2^{2n}})^t\to\vec{0}
$$
when $n\to\infty.$
On the other hand, by Remark  \ref{rembasic}, $A_{2n}$ has the following form( see also \cite{DGR}):
$$
A_{2n}=
\begin{pmatrix}
a_n& b_n\\
-b_n& c_n
\end{pmatrix}.
$$
Thus $a_n-b_n=\sqrt{2}\phi_{2^{2n}+1}$ and $-b_n-c_n=\sqrt{2}\phi_{2^{2n}}$.
Consequently
$$
t_{2n}(E)=a_n+c_n=\sqrt{2}(\phi_{2^{2n}+1}-\phi_{2^{2n}})\to 0.
$$
By the recurrence relation of $t_n$, we get
$$
t_{2n+1}(E)=2-\frac{2-t_{2n+2}(E)}{t_{2n}^2(E)}\to -\infty.
$$
By Lemma \ref{char-II-III}, $E$ is a type-II energy.
\hfill $\Box$

\bigskip

\noindent {\bf Proof of Theorem \ref{no-point}}. \
Assume on the contrary that  $E$ is a point spectrum of $H_\lambda$,
then there exists nonzero $\phi\in\ell^2(\Z)$ such that $H_\lambda\phi=E\phi.$
Then by the above lemma, $E$ is a type-II enenrgy.
However this will contradict with Theorem \ref{main-type-II},
since by that theorem, no $\ell^2$ solution is possible for $E\in \Sigma_{II}.$
\hfill $\Box$

\smallskip

\noindent{\bf Acknowledgement}. The authors  thank Professor Wen Zhiying and Professor Rao Hui for helpful discussions. Part of the work was done when the second author visited the Chinese University of Hong Kong in July 2015, he thanks CUHK for their hospitality.   Liu was  supported by the National Natural Science Foundation of China, No. 11371055 and No. 11571030.  Qu was supported by the National Natural Science Foundation of China, No. 11201256, No. 11371055 and No. 11431007.


\end{document}